\documentclass[aps,pra,twocolumn,superscriptaddress,nofootinbibt,longbibliography]{revtex4-2}
\pdfoutput=1
\usepackage[utf8]{inputenc}
\usepackage[english]{babel}
\usepackage[T1]{fontenc}
\usepackage{amsmath}
\usepackage{hyperref}

\usepackage{tikz}
\usepackage{lipsum}

\usepackage{amssymb}
\usepackage{amsfonts}
\usepackage{epsfig,pstricks,graphicx}
\usepackage{bm}
\usepackage{physics}
\usepackage{amsthm}

\usepackage{color}
\usepackage{array}
\def\id{{\rm 1\kern-.22em l}}

\newtheorem{lem}{Lemma}
\newtheorem{deff}{Definition}
\newtheorem{teor}{Theorem}

\begin{document}

\title{Duality of averaging of  quantum states over arbitrary symmetry groups revealing Schur-Weyl duality}

\author{Marcin Markiewicz}
\affiliation{International Centre for Theory of Quantum Technologies (ICTQT),
University of Gdansk, 80-308 Gdansk, Poland}
\author{Janusz Przewocki}
\affiliation{Institute of Mathematics, University of Gdansk, Wita Stwosza 57, 80-308 Gdansk, Poland}
\email{marcinm495@gmail.com, jprzew@mat.ug.edu.pl}

\begin{abstract}
It is a well-established fact in quantum information theory, that uniform averaging over the collective action of a unitary group on a multipartite quantum state projects the state to a form equivalent to a permutation operator of the subsystems. Hence states equivalent to permutation operators are untouched by collective unitary noise. A trivial observation shows that uniform averaging over permutation operators projects the state into a form with block-diagonal structure equivalent to the one of the collective action of the unitary group. We introduce a name for this property: \textit{duality of averaging}. The mathematical reason behind this duality is the fact that the collective action of the unitary group on the tensor product state space of a multipartite quantum system and the action of the permutation operations are mutual commutants when treated as matrix algebras. Such pairs of matrix algebras are known as \textit{dual reductive pairs}. In this work we show, that in the case of finite dimensional quantum systems such duality of averaging holds for \textit{any} pairs of symmetry groups being dual reductive pairs, regardless of whether they are compact or not, as long as the averaging operation is defined via iterated integral over the Cartan decomposition of the group action. Although our result is very general, we focus much attention on the concrete  example of a dual reductive pair consisting of collective action of special linear matrices and permutation operations, which physically corresponds to averaging multipartite quantum states over non-unitary SLOCC-type (Stochastic Local Operations and Classical Communication) operations. In this context we show, that \textit{noiseless subsystems} known from collective unitary averaging persist in the case of SLOCC averaging in a conditional way: upon postselection to specific invariant subspaces.
\end{abstract}

\maketitle

\section{Introduction}

\subsection{Schur-Weyl duality}
Schur-Weyl duality \cite{Goodman09, Brundan08, Dipper08, Doty09, Marvian14, Zhang15, Gross21} and the concept of \textit{noiseless subsystems} \cite{Zanardi97, Knill00, Kempe01, ZanardiVirt01, Bartlett03a, Bartlett07, Migdal11} are two faces of the same coin. Schur-Weyl duality  introduces pairs of matrix groups (called \textit{dual reductive pairs)}, which, treated as matrix algebras, maximally mutually commute, therefore quantum states spanned by operators from one of the elements of the dual pair are untouched by the quantum evolution encoded by operators from the second element of the pair \cite{Marvian14}. There is however one important assumption in the above story: both groups have to  possess finite-dimensional unitary representations. Indeed, if $\mathbb U$ is a subgroup of the unitary group $\operatorname{U}(d)$, treated here as a matrix algebra, and $\mathbb S$ is a commutant of $\mathbb U$, then any quantum state $\rho$, which is in the span of $\mathbb S$, by definition (and linearity) commutes with arbitrary element $V\in \mathbb U$, and therefore is untouched by the action of $V$:
\begin{equation}
    \label{noiselessUnit}
    V\rho V^{\dagger}=\rho VV^{\dagger}=\rho.
\end{equation}
The canonical example of a dual reductive pair contains the following two matrix algebras:
\begin{itemize}
    \item algebra spanned by the collective action $U^{\otimes t}$ of a unitary group $\textrm U(d)$ on a complex tensor space $\mathbb (\mathbb C^d)^{\otimes t}$,
    \item algebra generated by unitary representation of permutation operators on the tensor space $\mathbb (\mathbb C^d)^{\otimes t}$.
\end{itemize}
In standard quantum information processing terms,  first algebra represents \textit{collective unitary noise} on a system of $t$ $d$-level quantum systems (qudits), whereas the second represents swapping of subsystems. The moral of the story is that quantum states which are equivalent to permutation operators are untouched by the collective unitary noise. 

In the literature one can find other examples of dual reductive pairs, in which one of the terms is represented by collective action of some subgroup of the unitary group, whereas the second is some discrete or continuous group containing the symmetric group as a subgroup. The examples are:
\begin{itemize}
    \item Schur-Weyl duality for collective action of the Clifford unitary operators \cite{Gross21},
    \item Schur-Weyl duality for \textit{gauge groups} \cite{Marvian14}, which are defined as centralisers of some subgroup of a unitary group with respect to the unitary group.
\end{itemize}
All of the above examples share the same interpretation in the context of noiseless subsystems.

\subsection{Duality of averaging}
The idea of \textit{averaging} quantum states over symmetry operations appears in numerous applications within the entire field of quantum information science. It is a crucial
tool in analysing noisy quantum communication, both in the non-relativistic \cite{Bartlett03a,  Bartlett07, Migdal11} as well as relativistic context  \cite{Bartlett05,  Dragan15, Piani16}. Averaging over symmetry operations plays also an important role in developing randomised quantum information processing protocols \cite{Dankert05, Hunter19} and quantum algorithms \cite{Brandao13}. 
Schur-Weyl duality for unitary groups and their subgroups has a crucial impact on the idea of collective uniform averaging  of quantum states with respect to these groups (called \textit{twirling} in the quantum information context \cite{Bartlett03a, Bartlett07}). Namely such a collective averaging projects a quantum state into the subspace spanned by the commutant (second counterpart of the dual reductive pair) \cite{Bartlett03a, Bartlett07}.
In more details this aspect of Schur-Weyl duality is related with the group-representation-theoretic property of a dual reductive pair, namely that irreducible representations of each elements of a pair are in one-to-one correspondence, and that their joint action has a unique block-diagonal decomposition induced  by this correspondence.
 For example uniform collective unitary averaging of a multipartite finite-dimensional quantum state projects the state into the block-diagonal form  consisting of  projections of the state onto subspaces, which are irreducible with respect to the action of  the permutation group. An obvious, though not already found by the authors anywhere in the literature, consequence of the Schur-Weyl duality for uniform averaging, is that the same holds if one averages quantum state uniformly over the second factor in the dual reductive pair, namely the state is projected onto irreducible subspaces with respect to the first factor. We call this property \textit{duality of averaging with respect to dual reductive pair}. As a simple example, uniform averaging of a multipartite quantum state over symmetric group projects the state into a block-diagonal form consisting of projections of the state into subspaces irreducible with respect to the collective action of the unitary group.
 
 The open problem, which we state in this work, is the question, whether described \textit{duality of averaging} holds if at least one component of the reductive pair is a non-compact symmetry group. An emblematic example of such a pair is the collective action of a special linear group $\textrm{SL}(d,\mathbb C)$, which physically represents collective SLOCC-type operations (Stochastic Local Operations and Classical Communication) \cite{Dur00, Bennett00, Verstraete02, Donald02, Avron09, Migdal13, Sawicki14, Jarvis14, Zhang16}, and the symmetric group \cite{Goodman09, Etingof10}. On the one hand there seem to appear two crucial obstacles. The first one is that due to non-unitarity of the SLOCC operations, the quantum state, which commutes with the SLOCC operation still seems to be seriously affected by this transformation. Indeed, if $L\in \textrm{SL}(d,\mathbb C)$ and $[\rho, L]=0$ we have:
 \begin{equation}
    \label{noiselessSLOCC}
    L\rho L^{\dagger}=\rho LL^{\dagger}\neq \rho.
\end{equation}
 In general the operator $LL^{\dagger}$ is not proportional to identity, therefore the problem is more serious than just a normalisation issue. The second problem is that averaging over SLOCC operations, due to non-compact character of the transformations, cannot be performed in a uniform way, and a suppressing integration measure has to be introduced in order to assure convergent result. On the other hand, as shown in our previous work \cite{Markiewicz21}, an important aspect of unitary averaging, namely existence of \textit{finite averaging sets} (known in the quantum information community as unitary $t$-designs \cite{Dankert05, Gross07, Dankert09, Roy09}) persists in the process of averaging over non-compact $\textrm{SL}(d,\mathbb C)$ group. 
 
 \subsection{Main results}
 In this work we show, that the \textit{duality of averaging} indeed persists in the context of averaging over SLOCC operations with a slight modification, namely such a process of averaging projects a quantum state into block-diagonal form consisting of projections onto subspaces irreducible with respect to the symmetric group, but with additional non-trivial weights, which correspond to the process of averaging over the non-compact part of the SLOCC operation. These weights explicitly depend on the assumed measure of integration over the non-compact part. Therefore we can say, that globally there are no noiseless subsystems under the collective action of the SLOCC operations, however we can conditionally restore them by restricting to a single irreducible subspace and performing projection to this subspace, which succeeds with probability specified by the additional weight factors.
 
 Further on we prove a general statement that the duality of averaging in the defined sense persists for \textit{any} dual reductive pair regardless of the compactness of the symmetry group under consideration, as long as one defines the generalised twirling operation using the Cartan decomposition of the group operation.
 
 \subsection{Outline}
 The paper is organised as follows. In Section \ref{Sec:SWIntro} we introduce the concept of duality of averaging for the well-known unitary twirling (with all the technical details moved to the Appendices), in Section \ref{Sec:AVSLOCCDef} we show how one can define and understand averaging quantum states over SLOCC operations, in Section \ref{Sec:ExactSLOCC} we derive the closed analytical form of the SLOCC-twirling map. Finally we prove the general result on duality of averaging with respect to any dual reductive pair in Section \ref{Sec:Conj} and present final remarks in Section \ref{Sec:Conc}.

\section{Schur-Weyl duality and duality of averaging over collective  unitary operations}
\label{Sec:SWIntro}

Schur-Weyl duality for the unitary group $\textrm U(d)$ and the symmetric group $\textrm S_t$ over the set of $t$ elements describes an interplay between the collective action of the unitary group and the symmetric group on the tensor product space $(\mathbb C^d)^{\otimes t}$. The collective action of the unitary group is specified by operators  $U^{\otimes t}, U \in \textrm{U}(d)$, whereas the action of the symmetric group $\textrm{S}_t$ is defined as a permutation of the factors in $(\mathbb C^d)^{\otimes t}$. Let $$\ket{v}=\ket{v_{1}}\otimes\cdots\otimes\ket{v_{t}}$$ be arbitrary product vector in $(\mathbb C^d)^{\otimes t}$. Then any permutation $p\in \textrm{S}_t$ acts on $\ket{v}$ as follows:
\begin{equation}
    p\left(\ket{v}\right)=\ket{v_{p^{-1}(1)}}\otimes\cdots\otimes\ket{v_{p^{-1}(t)}}.
\end{equation}
 For further considerations it is necessary to describe the above action by a matrix transformation. Therefore we introduce the following matrix representation of the symmetric group on the complex vector space $(\mathbb C^d)^{\otimes t}$ known as \textit{tensor permutation operators} \cite{Christian06, Christian07}. To every permutation $p\in \textrm{S}_t$ we associate an orthogonal matrix $O_p\in \textrm{O}(d^t)$ as follows. Let $\tau(p)=\tau_1\ldots \tau_{n-1}=(i_1,i_2)(i_2,i_3)\ldots(i_{n-1}i_{n})$ be a (in general non-unique) decomposition of $p$ into transpositions. Then the matrix $O_p$ reads:
\begin{equation}
\label{ORep}
    O_p=O_{\tau_1}\ldots O_{\tau_{n-1}},
\end{equation}
in which the orthogonal transposition matrices are defined as \cite{Christian06, Christian07}:
\begin{equation}
    O_{(i_k, i_{k+1})}=\frac{1}{d}\id_{d^t}+\frac{1}{2}\sum_{j=1}^{d^2-1}\id\otimes\ldots\otimes\gamma_j\otimes\ldots\otimes\gamma_j\otimes\ldots\otimes\id,
\end{equation}
in which $\gamma_j$ are generalised Gell-Mann matrices (Hermitian generators of $\mathfrak u(d)$ algebra) placed at positions $i_k$ and $i_{k+1}$ in the tensor product. Since $\tau(p)\tau(q)=\tau(pq)$, we have due to \eqref{ORep} $O_pO_q=O_{pq}$, and therefore the mapping $p\mapsto O_p$ is indeed a group representation.

The actions on the space $(\mathbb C^d)^{\otimes t}$ of the permutation group and of the unitary group mutually commute in a \textit{maximal sense}, namely the matrix algebras generated by $\{U^{\otimes t}\}_{U\in \textrm{U}(d)}$ and $\{O_p\}_{p\in \textrm{S}_t}$ are mutual commutants in the algebra of endomorphisms of $(\mathbb C^d)^{\otimes t}$. This implies, that there exists a basis of $(\mathbb C^d)^{\otimes t}$, called \textit{Schur basis}, which via outer product of vectors generates two operator bases block-diagonalising both matrix algebras $\{U^{\otimes t}\}_{U\in \textrm{U}(d)}$ and $\{O_p\}_{p\in  \textrm{S}_t}$. In order to distinguish between the vector basis and the two operator bases, we will refer to the basis on $(\mathbb C^d)^{\otimes t}$ as \textit{Schur vector basis}, whereas the two operator bases for operators acting on this space will be referred to as \textit{Schur operator bases}.

Let us first introduce the structure of the Schur vector basis, the
elements of which will be enumerated by three indices: $\ket{i,m,\lambda}$, in which $i$ numbers irreducible representations of unitary and symmetric groups up to isomorphism, whereas the remaining indices $m$ and $\lambda$ describe further structure of each $i$-th subspace, which will be introduced later on.
From representation-theoretic point of view this means that the joint action of both groups is reducible on $(\mathbb C^d)^{\otimes t}$ and that moreover it decomposes into irreducible actions of both groups, which are in one-to-one correspondence. The detailed construction of the Schur vector basis is described in the Appendix \ref{App:SWDuality}, here we just define the crucial tools needed for further considerations.

The \textit{inequivalent} irreducible representations of both groups: the unitary group $\textrm U(d)$ and the symmetric group $\textrm S_t$ on the tensor space   $(\mathbb C^d)^{\otimes t}$ are labeled by all inequivalent \emph{Young diagrams} with $t$ boxes and at most $d$ rows, which correspond to all inequivalent partitions of $t$ into the sum of $d$ positive integers \cite{Tung}. Let us label these representations with the index $i$. As a side remark we point out that there exists no closed analytic formula for the range of index $i$, namely for the number of inequivalent partitions with fixed number of terms \cite{Oruc16}. Therefore we would not specify the range of the index $i$.
Each $i$-th  subspace of $(\mathbb C^d)^{\otimes t}$ irreducible with respect to the joint action of $U^{\otimes t}$ and $O_p$, contains in general multiple subspaces, which are irreducible with respect to the separate action of symmetric or collective unitary operators. These subspaces correspond to \textit{equivalent but distinct} irreducible representations of the corresponding groups, and therefore each $i$-th subspace can be further organised in the following way:
\begin{eqnarray}
    \label{SchurBasis0}
     L^i_1:\,\,&&\ket{i,1,1}\,\,\,\,\,\,\,\,\,\ket{i,2,1}\,\,\,\,\,\,\ldots\,\,\,\,\,\,\ket{i,D^i_L,1} \nonumber\\
     L^i_2:\,\,&&\ket{i,1,2}\,\,\,\,\,\,\,\,\,\ket{i,2,2}\,\,\,\,\,\,\ldots\,\,\,\,\,\,\ket{i,D^i_L,2} \nonumber\\
     && \ldots\ldots\ldots\ldots\ldots\ldots\ldots\ldots\ldots\ldots\ldots\ldots \nonumber\\
       L^i_{D^i_V}:\,\,&&\underbrace{\ket{i,1,D^i_V}}_{V^i_1}\,\,\underbrace{\ket{i,2,D^i_V}}_{V^i_2}\,\,\,\ldots\,\,\,\underbrace{\ket{i,D^i_L,D^i_V}}_{V^i_{D^i_L}}.\nonumber\\
    \end{eqnarray}
In the above table each vector of the form $\ket{i,m,\lambda}$ represents one element of the Schur basis. Subspaces $L^i_{\lambda}$, spanned by vectors  $\{\ket{i,m,\lambda}\}_{m=1}^{D^i_L}$ occurring in each row, are invariant and irreducible under the collective action of the unitary group $U^{\otimes t}$, and correspond to equivalent irreducible representations of the unitary group $\textrm U(d)$ of dimension $D^i_L$ (explicit formulae for the dimensions can be found in the Appendix).
Similarly the subspaces $V^i_m$ spanned by vectors $\{\ket{i,m,\lambda}\}_{\lambda=1}^{D^i_V}$ appearing in each column are invariant and irreducible under the action of the symmetric group $\textrm S_t$, and correspond to equivalent $D^i_V$ dimensional irreducible representations of the symmetric group (they are called \textit{Specht moduli} in the representation theory \cite{Etingof10}). It can be seen that the subspaces $L^i_{\lambda}$ and $V^i_m$ play the role of \textit{multiplicity spaces} for the corresponding representations: subspaces $V^i_m$ are multiplicity spaces for representations of the unitary group, whereas subspaces $L^i_{\lambda}$ are multiplicity spaces for representations of the symmetric group \cite{Bartlett07}. Direct sum:
\begin{equation}
    \label{pspace}
P_i \equiv \bigoplus_\lambda L^i_\lambda = \bigoplus_m V^i_m
\end{equation}
contains all subrepresentations of $U^{\otimes t}$ and $S_\textrm t$ up to the same isomorphism type.

Before going further let us illustrate the described Schur vector basis decomposition by an example of the tensor space $(\mathbb C^2)^{\otimes 3}$, which physically represents a state space of three qubits (two-level systems). Under the joint action of $\textrm U(2)^{\otimes 3}$ and $S_3$ the space $(\mathbb C^2)^{\otimes 3}$ decomposes into two irreducible subspaces corresponding to inequivalent irreducible representations of the unitary group $\textrm U(2)$ and $S_3$:
\begin{itemize}
    \item $i=1$ subspace, corresponding to a Young diagram with three boxes places in one row: $4$-dimensional subspace of symmetric tensors, spanned by Schur basis vectors:
   \begin{equation}
        L^1_{1}:\,\,\underbrace{\ket{1,1,1}}_{V^1_1}\,\,\underbrace{\ket{1,2,1}}_{V^1_2}\,\,\underbrace{\ket{1,3,1}}_{V^1_3}\,\,\underbrace{\ket{1,4,1}}_{V^1_{4}},
   \end{equation}
    on which there acts a one-dimensional trivial representation of the symmetric group $S_3$ (hence $D^1_V=1$), and $4$-dimensional representation of the $\textrm U(2)$ group, namely the spin-$\tfrac{3}{2}$ representation (hence $D^1_L=4$). Note that the entire $i=1$ subspace is irreducibly invariant under the action of $U^{\otimes 3}$, whereas each one-dimensional subspace corresponding to each element of Schur vector basis is invariant under symmetric group $S_3$ (since each element of Schur vector basis is here fully symmetric when expressed in computational basis).
    \item $i=2$ subspace, corresponding to a Young diagram with two boxes in first row and one box in the second: $4$-dimensional subspace of mixed-symmetry tensors, spanned by Schur basis vectors:
 \begin{eqnarray}
     L^2_1:\,\,&&\ket{2,1,1}\,\,\ket{2,2,1} \nonumber\\
       L^2_{2}:\,\,&&\underbrace{\ket{2,1,2}}_{V^2_1}\,\,\underbrace{\ket{2,2,2}}_{V^2_2},
    \end{eqnarray}
    on which there act two equivalent irreducible two-dimensional representations of the unitary group $\textrm U(2)$ (spin-$\tfrac{1}{2}$ representations) with the corresponding invariant subspaces denoted by $L^2_1$ and $L^2_2$, and two equivalent two-dimensional representations of the symmetric group $S_3$ with the corresponding invariant subspaces $V^2_1$ and $V^2_2$.
\end{itemize}

In most common approach to Schur-Weyl duality one treats the Schur basis vectors $\ket{i,m,\lambda}$ as \textit{virtual tensor products} $\ket{m}_i\otimes\ket{\lambda}_i$, which correspond to tensor products of vectors from \textit{virtual subspaces} $\{\ket{m}_i\}_{m=1}^{D^{i}_L}$ and $\{\ket{\lambda}_i\}_{\lambda=1}^{D^{i}_V}$ \cite{Zanardi97, ZanardiVirt01, Bartlett07}. Here we introduce other approach, based on operator bases build up from outer products of Schur vector basis elements. These operator bases are extremely useful when dealing with averaging maps. Let us start from the most general operator basis build up of Schur basis vectors:
\begin{equation}
    \label{FullPiBasis0}
    \hat\Pi_{ij}^{m_1\lambda_1 m_2\lambda_2}=\ket{i,m_1,\lambda_1}\bra{j,m_2,\lambda_2}.
\end{equation}
The above operators allow us to define two \textit{Schur operator bases}:
\begin{eqnarray}
\label{PiBasis0}
\hat\Pi^{\lambda_1\lambda_2}_i&=&\sum_{m=1}^{D^i_L}\hat\Pi_{ii}^{m\lambda_1 m\lambda_2},\nonumber\\
\hat\Pi^{m_1 m_2}_i&=&\sum_{\lambda=1}^{D^i_V}\hat\Pi_{ii}^{m_1\lambda m_2\lambda}.
\end{eqnarray}
The set $\{\hat\Pi^{\lambda_1\lambda_2}_i\}$ is a basis for operators, which leave invariant the subspaces $V^i_m$  irreducible under the action of the symmetric group, whereas the set $\{\hat\Pi^{m_1m_2}_i\}$ spans operators, which leave invariant the subspaces $L^i_{\lambda}$ irreducible under the collective action of the unitary group. Therefore the collective action of the unitary group $U^{\otimes t}$ as well as the orthogonal representation of the symmetric group $O_p$ have natural decompositions in terms of these operator bases:
\begin{equation}
    \label{Ut0}
    U^{\otimes t}=\sum_{i}\frac{1}{D^i_V}\sum_{m_1,m_2=1}^{D^i_L}\operatorname{Tr}\left(U^{\otimes t}\hat\Pi^{m_1 m_2\dagger}_i\right)\hat\Pi^{m_1 m_2}_i,
\end{equation}
and analogously:
\begin{equation}
    \label{Op0}
    O_{p}=\sum_{i}\frac{1}{D^i_L}\sum_{\lambda_1,\lambda_2=1}^{D^i_V}\operatorname{Tr}\left(O_p\hat\Pi^{\lambda_1 \lambda_2\dagger}_i\right)\hat\Pi^{\lambda_1 \lambda_2}_i.
\end{equation}
The fact that matrix algebras generated by two kinds of operators \eqref{Ut0} and \eqref{Op0} are mutual commutants is reflected by the commutativity of the above defined operator bases:
\begin{equation}
    \label{PiComm0}
    [\hat\Pi^{m_1 m_2}_i,\hat\Pi^{\lambda_1 \lambda_2}_j]=0.
\end{equation}
Let us now define two averaging operations with respect to both elements of the dual reductive pair:
\begin{equation}
\label{UnitTwirlDef0}
    \mathcal T_{\operatorname{U}}(\rho)=\int U^{\otimes t}\rho U^{\dagger \otimes t} dU, \,U \in \textrm U(d),
\end{equation}
\begin{equation}
    \label{SymTwirlDef0}
    \mathcal T_{\textrm{sym}}(\rho)=\frac{1}{t!}\sum_{p\in \operatorname{S}_t}O_{p}\rho O_p^T.
\end{equation}
The first operation, known as \textit{unitary twirling} \cite{Bartlett07}, represents averaging a quantum state $\rho$ of $t$ $d$-level subsystems over the collective action of the unitary group in a uniform way: $dU$ represents a normalised Haar measure on the unitary group. Such operation can have different physical interpretations, the two most common are the following:
\begin{itemize}
    \item  multipartite quantum system under consideration is subject to collective local unitary noise \cite{Marvian14, Markiewicz17},
    \item multipartite quantum system is sent between two observers that do not share a common reference frame \cite{Bartlett03a, Bartlett07}.
\end{itemize}
The second operation, which we call \textit{symmetric twirling} is non-local with respect to subsystems and has possible applications in the theory of randomised quantum protocols \cite{Holger17}. 

The fact that both twirling operations \eqref{UnitTwirlDef0} and \eqref{SymTwirlDef0} represent averaging with respect to symmetry groups that belong to a dual reductive pair implies the following duality of averaging:
\begin{equation}
    \label{fullUnitTwirlMap0}
    \mathcal T_{\operatorname{U}}(\rho)=\sum_{k}\frac{1}{D^k_L}\sum_{\lambda_1\lambda_2=1}^{D^k_V}\operatorname{Tr}\left(\rho\hat\Pi_{k}^{\lambda_1\lambda_2\dagger}\right)\hat\Pi_{k}^{\lambda_1\lambda_2}.
\end{equation}
\begin{equation}
    \label{fullSymTwirlMap0}
    \mathcal T_{\textrm{sym}}(\rho)=\sum_{k}\frac{1}{D^k_V}\sum_{m_1,m_2=1}^{D^k_L}\operatorname{Tr}\left(\rho\hat\Pi^{m_1 m_2\dagger}_k\right)\hat\Pi^{m_1 m_2}_k.
\end{equation}
We can see that closed analytical forms of both maps are in the same relation as operators \eqref{Op0} and \eqref{Ut0}:
\begin{itemize}
    \item unitary twirling $\mathcal T_{\operatorname{U}}(\rho)$ projects a quantum state into block-diagonal form (with respect to block index $k$) consisting of subspaces  which are invariant and irreducible under the action of the symmetric group;
    \item symmetric twirling $\mathcal T_{\textrm{sym}}(\rho)$ projects a quantum state into block-diagonal form consisting of subspaces which are invariant and irreducible under the collective action of the unitary group.
\end{itemize}
Both analytical formulas \eqref{fullUnitTwirlMap0} and \eqref{fullSymTwirlMap0} are simple consequences of Schur's Lemma (see Appendix \ref{App:SLemmas}). Their explicit proofs are however a bit technical (though standard),  we present them for completeness in Appendices \ref{App:UnitTwirl} and \ref{App:SymTwirl}. Both twirling operations \eqref{fullUnitTwirlMap0} and \eqref{fullSymTwirlMap0} are idempotent: $\mathcal T_{\operatorname{U}}(\mathcal{T}_{\operatorname{U}}(\rho))=\mathcal T_{\operatorname{U}}(\rho)$,  $\mathcal T_{\textrm{sym}}(\mathcal{T_{\textrm{sym}}}(\rho))=\mathcal T_{\textrm{sym}}(\rho)$ (see Appendices \ref{App:UnitTwirl} and \ref{App:SymTwirl}), which implies that states of the form:
\begin{eqnarray}
    \label{dfsstate}
    \rho_{\textrm{U}}&=&\sum_{k}\frac{1}{D^k_L}\sum_{\lambda_1\lambda_2=1}^{D^k_V}\rho^{k}_{\lambda_1\lambda_2}\hat\Pi_{k}^{\lambda_1\lambda_2},\nonumber\\
    \rho_\textrm{S}&=&\sum_{k}\frac{1}{D^k_V}\sum_{m_1,m_2=1}^{D^k_L}\rho_{m_1 m_2}^k\hat\Pi^{m_1 m_2}_k,
\end{eqnarray}
span \textit{noiseless subsystems} with respect to the twirling operations:
\begin{eqnarray}
\label{unitdfs}
&&\mathcal T_{\operatorname{U}}( \rho_{\textrm{U}})= \rho_{\textrm{U}},\nonumber\\
&&\mathcal T_{\textrm{sym}}( \rho_{\textrm{S}})= \rho_{\textrm{S}}.
\end{eqnarray}

\section{Averaging multipartite quantum states over SLOCC operations}
\label{Sec:AVSLOCCDef}

In this section we introduce the concept of 
averaging multipartite quantum states of a finite dimension $d$ over the set of the most general quantum operations, namely Stochastic Local Operations and Classical Communication (SLOCC) \cite{Dur00, Bennett00, Verstraete02, Donald02, Avron09, Migdal13, Sawicki14, Jarvis14, Zhang16}. Mathematically, these operations are described by the following map \cite{Avron09}:
\begin{equation}
\label{SLOCCMap}
    \mathcal S(\rho) = \bigotimes_{i=1}^t L_i\rho\bigotimes_{i=1}^t L_i^{\dagger},
\end{equation}
where the matrices $L_i$ are normalized special linear $\textrm{SL}(d,\mathbb C)$ matrices, namely:
\begin{equation}
    L=\frac{M}{\|M\|}, \, M \in \textrm{SL}(d,\mathbb C).
\end{equation}

\begin{lem}
\label{lem:trace}
The map $\mathcal S$ preserves the positivity of the matrix $\rho$, and is trace non-increasing.
\end{lem}
\begin{proof}
The fact that $\mathcal S$ preserves positivity is a simple consequence of the tensor product commuting with Hermitian transpose. The proof that $\mathcal S$ is trace non-increasing requires some more computation and is presented for completeness in the Appendix \ref{App:LemTrace}.\end{proof}
Both the above properties stated in Lemma \ref{lem:trace} assure that the map \eqref{SLOCCMap} is a well-defined quantum operation. In general the trace of the matrix $\mathcal S(\rho)$ is strictly less than one. Such a state can be physically interpreted as a properly normalized quantum state $\mathcal S(\rho)/\operatorname{Tr}(\mathcal S(\rho))$, which is prepared with success probability $\operatorname{Tr}(\mathcal S(\rho))$. The acceptance of a preparation of this state demands a classical communication between the local observers performing local operations $L_i$, hence the name of the class. 

Let us now introduce the concept of collective averaging over SLOCC operations, which will be done in analogy to the averaging over unitary operations. The term \emph{collective} refers to the fact that all the local operations are equal, therefore we consider the maps:
\begin{equation}
    \mathcal S_\textrm{C}(\rho)=L^{\otimes t}\rho L^{\dagger\otimes t}.
\end{equation}
In the case of unitary operations there is a unique way of defining collective averaging, see formula \eqref{UnitTwirlDef0}.
Generalization of the unitary twirling operation to the SLOCC case meets two problems. The first one is that the operation $S_{\textrm C}(\rho)$ does not preserve the trace, the second one is that the group $\textrm{SL}(d,\mathbb C)$ is non compact, and therefore it is not possible to average over it in a uniform way. 
Let us first focus on the second problem, and further, after defining our averaging process we will tackle the problem with normalisation. Our aim is to define the averaging process over SLOCC operations in a way which preserves as much uniformity as possible. Averaging by integrating over Haar measure on $\textrm{SL}(d,\mathbb C)$ group is not possible, since this measure is not finite and the integral would be divergent. What is more, the group $\textrm{SL}(d,\mathbb C)$ is not an \textit{amenable group} \cite{Bondar81}, which means that it is not possible to define \textit{any} functional on bounded functions defined on this group, which would be: normalised, non-negative definite and left or right invariant with respect to the group action. Therefore any averaging procedure for functions defined on the group $\textrm{SL}(d,\mathbb C)$ must necessarily involve some sort of non-uniformity. In order to minimize it we should try to decompose the group action in order to separate the non-compact part. This idea is realised by the Cartan decomposition for the group $\textrm{SL}(d,\mathbb C)$, which is in this case equivalent to a singular value decomposition (SVD) on the level of matrices. In simple terms this means that an arbitrary matrix $M\in \textrm{SL}(d,\mathbb C)$ can be represented as a product $M=KAK'$, in which matrices $K,K'\in \textrm{SU}(d)$, whereas $A$ is some diagonal traceless real matrix \cite{Knapp, Helgason}, which represents some abelian group $\mathbb A$. The special unitary group $\textrm{SU}(d)$ plays the role of a \textit{maximal compact subgroup} of $\textrm{SL}(d,\mathbb C)$, whereas an abelian group $\mathbb A$ is its \textit{maximal abelian subgroup}. However it is important to note that the direct product group  $\textrm{SU}(d)\times \mathbb A\times \textrm{SU}(d)$ is \textit{not} homomorphic to $\textrm{SL}(d,\mathbb C)$. What is more, the mapping:
\begin{equation}
    \label{kakmap}
    \textrm{SU}(d)\times \mathbb A\times \textrm{SU}(d)\mapsto\textrm{SL}(d,\mathbb C),
\end{equation}
realised by $(K, A, K')\mapsto M$ is \textit{not} a diffeomorphism between the group manifolds  $\textrm{SU}(d)\times \mathbb A\times \textrm{SU}(d)$ and 
$\textrm{SL}(d,\mathbb C)$ due to a mismatch of dimensions. Namely the group manifold corresponding to the direct product $\textrm{SU}(d)\times \mathbb A\times \textrm{SU}(d)$ has higher dimension that the manifold corresponding to $\textrm{SL}(d,\mathbb C)$. For example in the case of $\textrm{SL}(2,\mathbb C)$, $\textrm{SU(2)}$ is parametrised by $3$ real parameters, whereas $\mathbb A$ is one-dimensional and can be represented by matrices of the form: 
\begin{equation}
\label{ASL2C}
\mathbb A_2 = \left\{ \left( \begin{array}{cc} x & 0 \\ 0 & x^{-1}  \end{array}  \right), \,\,x \geq 1  \right\}.
\end{equation}
Therefore the product manifold $\textrm{SU}(2)\times \mathbb A_2\times \textrm{SU}(2)$ has real dimension equal to $7$, whereas the group manifold of $\textrm{SL}(2,\mathbb C)$ is $6$-dimensional. Hence on the one hand Cartan decomposition allows us to decompose the SLOCC-type operations into the maximally compact and non-compact parts, on the other hand the mismatch of dimensions in the mapping \eqref{kakmap} complicates the averaging process. Nevertheless, since any smooth function $f$ on $\textrm{SL}(d,\mathbb C)$ can be expressed by a smooth function $\tilde f$ on $\textrm{SU}(d)\times \mathbb A\times \textrm{SU}(d)$, we can \textit{define} an integral over  $\textrm{SL}(d,\mathbb C)$ by an \textit{iterated integral} over the product manifold corresponding to $\textrm{SU}(d)\times \mathbb A\times \textrm{SU}(d)$:
\begin{equation}
\label{SLint}
   \int_{\textrm{SL}(d,\mathbb C)}fd\mu(M)\equiv \int_{\textrm{SU}(d)\times \mathbb A\times \textrm{SU}(d)}\tilde fd\mu(K)d\mu(A)d\mu(K').
\end{equation}
The above definition should be understood as follows: by fixing a measure $d\mu(K)$ on the group manifold of $\textrm{SU}(d)$ and $d\mu(A)$ on the group manifold of $\mathbb A$ we \textit{define} a measure $d\mu(M)$ on the group manifold of $\textrm{SL}(d,\mathbb C)$ and therefore the right-hand-side integral can be treated as an integral over $\textrm{SL}(d,\mathbb C)$. Therefore we propose the following definition:
\begin{deff}
\label{def1}
Generalisation of the unitary collective twirling to the SLOCC twirling can be defined by the following map:
\begin{eqnarray}
  \mathcal T_{\operatorname{SL}}(\rho)=&&\int \left(\frac{KAK'}{||KAK'||}\right)^{\otimes t}\rho \left(\frac{KAK'}{||KAK'||}\right)^{\dagger\otimes t}\nonumber\\
    &&d\mu(K)d\mu(A)d\mu(K'),
    \label{mainTwirlSL1}
\end{eqnarray}
in which in the iterated integral the measure $d\mu(K)$ on the group manifold of $\textrm{SU(d)}$ is assumed to be the normalised Haar measure, whereas the measure $d\mu(A)$ is an arbitrary normalised measure on the group manifold of $\mathbb A$.
\end{deff}
Due to \eqref{SLint} the map \eqref{mainTwirlSL1} can be equivalently represented as:
\begin{equation}
    \mathcal T_{\textrm SL}(\rho)=\int \left(\frac{M}{||M||}\right)^{\otimes t}\rho \left(\frac{M^{\dagger}}{||M||}\right)^{\otimes t} d\mu(M),
    \label{mainTwirlSL}
\end{equation}
 in which $M\in \textrm{SL}(d,\mathbb C)$ and the measure $d\mu(M)$ on the group manifold of $\textrm{SL}(d,\mathbb C)$ is defined via relation \eqref{SLint}. Proposed definition assures as much uniformity of averaging over SLOCC operations as possible, namely the averaging over compact components of SLOCC operations is performed uniformly according to a Haar measure on the components. All the non-uniformity is shifted to averaging over non-compact component consisting of matrices $A\in \mathbb A$, interpreted physically as \textit{filtering operations} \cite{Kent99, Verstraete01, Avron09}. Very similar definition has been suggested in our previous work on construction of finite averaging sets for the group $\textrm{SL}(d,\mathbb C)$ \cite{Markiewicz21}, however with an important difference. There we utilised relation \eqref{SLint} in the opposite direction, namely we first defined the (divergent) averaging process by the left-hand-side with respect to the Haar measure on $\textrm{SL}(d,\mathbb C)$, translated this Haar measure to the product measure in the right-hand-side iterated integral and finally added a suppressing factor in the component of the Haar measure on $\mathbb A$ in order to assure convergence. Here we do not impose \textit{any} restriction on the measure $d\mu(A)$ other than that it has to be normalised.
 A similar concept of averaging to Definition \ref{def1} has been proposed in \cite{Serafini07} for defining averaging of optical states over Gaussian operations represented by symplectic matrices from a non-compact symplectic group $\textrm{Sp}(d,\mathbb R)$.
 
 Let us now make sure that defined map \eqref{mainTwirlSL} is a properly defined quantum operation:
 \begin{lem}
 \label{lem:slgen}
 The map $\mathcal T_{\operatorname{SL}}$ preserves the positivity of the density matrix $\rho$, and is trace non-increasing.
 \end{lem}
 \begin{proof}
 Since  $\mathcal T_{\textrm{SL}}(\rho)=\int \mathcal S_{\textrm C}(\rho)d\mu(M)$, and according to Lemma \ref{lem:trace} $S_{\textrm C}$ preserves positivity, $\mathcal T_{\textrm{SL}}$ as a probabilistic mixture of  $S_{\textrm C}$ also preserves positivity. The fact that it is also trace non-increasing follows simply from:
\begin{eqnarray}
\Tr\left(\mathcal T_{\textrm{SL}}(\rho)\right)&=&\Tr\left(\int \mathcal S_{\textrm C}(\rho)d\mu(M)\right)\nonumber\\
&=&\int \Tr\left(\mathcal S_{\textrm C}(\rho)\right)d\mu(M)\leq\int d\mu(M)=1,\nonumber\\
\end{eqnarray}
in which the inequality follows from the fact that $\Tr\left(\mathcal S_{\textrm C}(\rho)\right)\leq 1$ according to Lemma \ref{lem:trace}.
 \end{proof}
 The second issue is that the above defined map involves mixing of subnormalized states, therefore its physical interpretation is not clear at first glance. In order to provide such interpretation let us rewrite the map \eqref{mainTwirlSL} as follows:
 \begin{equation}
    \mathcal T_{\textrm{SL}}(\rho)=\int p_M(\rho)\frac{M^{\otimes t}\rho M^{\dagger\otimes t}}{\operatorname{Tr}(M^{\otimes t}\rho M^{\dagger\otimes t})}d\mu(M),
\end{equation}
in which the probability $p_M(\rho)$ reads:
\begin{equation}
    \label{succprob}
    p_M(\rho)=\frac{\operatorname{Tr}(M^{\otimes t}\rho M^{\dagger\otimes t})}{\|M\|^{2t}},
\end{equation}
and can be interpreted as a success probability of performing the SLOCC operation \cite{Avron09}. Therefore the entire map $\mathcal T_{\textrm{SL}}$ represents preparation of  an ensamble of properly normalised states $\frac{M^{\otimes t}\rho M^{\dagger\otimes t}}{\operatorname{Tr}(M^{\otimes t}\rho M^{\dagger\otimes t})}$ produced by the SLOCC operations $M$, for $M$ drawn randomly from the group $\textrm{SL}(d,\mathbb C)$ according to a probabilistic measure $d\mu(M)$, \textit{weighted} according to the success probability $p_M(\rho)$ of the respective drawn SLOCC operation $M$.
The normalisation of the output state of the map \eqref{mainTwirlSL} $\Tr(\mathcal T_{\textrm{SL}}(\rho))$ represents average success probability of the SLOCC operations:
\begin{equation}
    \label{pmav}
    \langle p_M(\rho)\rangle=\Tr(\mathcal T_{\textrm{SL}}(\rho))=\int p_M(\rho)d\mu(M).
\end{equation}

\section{Exact form of the twirling map for SLOCC operations}
\label{Sec:ExactSLOCC}

In this section we will calculate closed analytic form of the proposed SLOCC twirling map \eqref{mainTwirlSL1}. 
Our derivation utilizes the closed form of  the unitary twirling map \eqref{fullUnitTwirlMap0}, which is derived in Appendix \ref{App:UnitTwirl}. The main tool used there is the application of Schur's Lemmas based on invariance of the unitary twirling procedure with respect to fixed left and right unitary rotations, see Appendix \ref{App:SLemmas} for details. In the case of non-unitary averaging such invariance no longer holds, therefore we have to exploit other techniques related with the structure of the Cartan decomposition in \eqref{mainTwirlSL1}, which we have successfully used in the context of finding finite averaging sets for $\textrm{SL}(2,\mathbb C)$ group \cite{Markiewicz21}. These techniques allow us to prove a factorisation of the SLOCC twirling map \eqref{mainTwirlSL1} into two terms, the first one being the unitary twirling map \eqref{UnitTwirlDef0} applied to the input state, whereas the second being the unitary twirling map applied solely to the non-unitary part of  SLOCC operations.

Let us start by rewriting the map \eqref{mainTwirlSL1} in the following way.
Since $||KAK'||=||A||$ (which follows from the fact that the Cartan decomposition for $\textrm{SL}(d,\mathbb C)$ matrices is equivalent to a singular value decomposition), we can introduce normalised filtering matrix $A_{\textrm{n}}=A/||A||$, which leads to the following form of the map:
\begin{eqnarray}
    \mathcal T_{\textrm{SL}}(\rho)=\int&& \left(KA_{\textrm{n}}K'\right)^{\otimes t}\rho \left(KA_{\textrm{n}}K'\right)^{\dagger\otimes t}\nonumber\\
    &&d\mu(K)d\mu(A)d\mu(K').
    \label{mainTwirlSL2}
\end{eqnarray}
Using distributivity of the tensor product with respect to the product of matrices the above map can be decomposed in the following form:
\begin{eqnarray}
    \mathcal T_{\textrm{SL}}(\rho)=\int&& K^{\otimes t}A_{\textrm{n}}^{\otimes t}\left(\int K'^{\otimes t}\rho K'^{\dagger\otimes t}d\mu(K')\right)\nonumber\\
    &&A_{\textrm{n}}^{\dagger\otimes t}K^{\dagger\otimes t}d\mu(K)d\mu(A).
    \label{mainTwirlSL30}
\end{eqnarray}
As can be seen in the above formula a twirling operation with respect to a special unitary group $\textrm{SU}(d)$ appeared in the middle integral. In our previous work \cite{Markiewicz21} we have shown that such a twirling operation is equivalent to a twirling map with respect to the  entire unitary group $\mathcal T_{\operatorname{U}}(\rho)$ \eqref{UnitTwirlDef0}, therefore the  map \eqref{mainTwirlSL30} can be simplified to:
\begin{equation}
    \mathcal T_{\textrm{SL}}(\rho)=\int K^{\otimes t}A_{\textrm{n}}^{\otimes t}\mathcal T_{\operatorname{U}}(\rho)
    A_{\textrm{n}}^{\dagger\otimes t}K^{\dagger\otimes t}d\mu(K)d\mu(A).
    \label{mainTwirlSL3}
\end{equation}
Now it turns out that the unitarily twirled state $\mathcal T_{\operatorname{U}}(\rho)$
commutes with all other operators in the integral due to the following lemma:
\begin{lem}
\label{lem1}
Let us take any invertible $d\times d$ matrix $\alpha$ and any $d^t\times d^t$ matrix $B$. Then the matrices $\alpha^{\otimes t}$ and $\mathcal T_{\operatorname{U}}(B)$ commute for every $t$ and $d$.
\end{lem}
\begin{proof}
The $t$-th tensor power of arbitrary invertible matrix $\alpha$ can be represented in a block-diagonal form equivalent to the one of $U^{\otimes t}$ (see formula \eqref{GLDecomp} in the Appendix \ref{App:SWDuality}):
\begin{equation}
\label{alphaBlockDiag}
    \alpha^{\otimes t}=\sum_k\frac{1}{D^k_V}\sum_{m_1,m_2=1}^{D^k_L}\alpha^k_{m_1m_2}\hat\Pi^{m_1m_2}_k,
\end{equation}
in which $\alpha^k_{m_1m_2}=\operatorname{Tr}\left(\alpha^{\otimes t}\hat\Pi^{m_1m_2\dagger}_k\right)$, whereas the 
unitarily twirled operator $B$ has the form \eqref{fullUnitTwirlMap0}:
\begin{equation}
    \mathcal T_{\operatorname{U}}(B)=\sum_k\frac{1}{D^k_L}\sum_{\lambda_1,\lambda_2=1}^{D^k_V}B^k_{\lambda_1\lambda_2}\hat\Pi^{\lambda_1\lambda_2}_k,
\end{equation}
in which $B^k_{\lambda_1\lambda_2}=\operatorname{Tr}\left(B\hat\Pi^{\lambda_1\lambda_2\dagger}_k\right)$. Due to commutativity of the projection operators $\hat\Pi^{m_1m_2}_k$ and $\hat\Pi^{\lambda_1\lambda_2}_l$, see \eqref{PiCommute} in Appendix \ref{App:SWDuality}, we have:
\begin{eqnarray}
\left[\alpha^{\otimes t},\mathcal T_{\operatorname{U}}(B)\right]&=&\left[\frac{1}{D^k_V}\alpha^k_{m_1m_2}\hat\Pi^{m_1m_2}_k,\frac{1}{D^l_L}B^l_{\lambda_1\lambda_2}\hat\Pi^{\lambda_1\lambda_2}_l\right]\nonumber\\
&=&\frac{1}{D^k_VD^l_L}\alpha^k_{m_1m_2}B^l_{\lambda_1\lambda_2}\left[\hat\Pi^{m_1m_2}_k,\hat\Pi^{\lambda_1\lambda_2}_l\right]=0,\nonumber\\
\end{eqnarray}
in which in the above we used Einstein's summation convention to simplify the notation.
\end{proof}
Due to Lemma \ref{lem1} the operator $\mathcal T_{\operatorname{U}}(\rho)$ commutes with both operators $K^{\otimes t}$ and $A_{\textrm{n}}^{\otimes t}$, therefore it can be factored out of the integral in the formula \eqref{mainTwirlSL3}:
\begin{equation}
    \mathcal T_{\textrm{SL}}(\rho)=\mathcal T_{\operatorname{U}}(\rho)\int K^{\otimes t}\left(\int A_n^{\otimes t}
    A_n^{\dagger\otimes t}d\mu(A)\right)K^{\dagger\otimes t}d\mu(K).
    \label{mainTwirlSL4}
\end{equation}
Let us now introduce a compact notation for the inner integral:
\begin{equation}
    \mathcal A=\int A_{\textrm{n}}^{\otimes t}
    A_{\textrm{n}}^{\dagger\otimes t}d\mu(A)=\int \left(A_{\textrm{n}}^{\otimes t}\right)^2
  d\mu(A),
    \label{Aop}
\end{equation}
in which the second equality comes from the fact that $A_{\textrm{n}}$ is a real diagonal matrix. Finally the map \eqref{mainTwirlSL4} can be presented in a  factorised form:
\begin{equation}
    \mathcal T_{\textrm{SL}}(\rho)=\mathcal T_{\operatorname{U}}(\rho)\mathcal T_{\operatorname{U}}(\mathcal A),
    \label{mainTwirlSL5}
\end{equation}
in which the term $\mathcal T_{\operatorname{U}}(\mathcal A)$ plays the role of a \textit{correction} due to non-unitary character of SLOCC averaging.
The operator $\mathcal T_{\operatorname{U}}(\mathcal A)$ has a very special form, which is explicitly provided by the following lemma:
\begin{lem}
\label{lem2}
Let $\alpha$ be arbitrary $d\times d$ complex invertible matrix. Then the operator $\mathcal T_{\operatorname{U}}(\alpha^{\otimes t})$ is fully diagonal in  \textit{Schur operator basis}:
\begin{equation}
\label{twirledProdOp}
\mathcal T_{\operatorname{U}}(\alpha^{\otimes t})=\sum_k \frac{1}{D^k}\alpha_k\hat\Pi_k,  
\end{equation}
in which the coefficients read: $\alpha_k=\Tr(\alpha^{\otimes t}\hat\Pi_k^\dagger)$, $\hat\Pi_k$ is a projector onto the $k$-th  subspace $P_k$ \eqref{pspace} and $D^k$ is the dimension of the subspace $P_k$: $D^k=D^k_LD^k_V$.
\end{lem}
\begin{proof}
See Appendix \ref{proof:lem2}.
\end{proof}
From Lemma \ref{lem2} it follows that:
\begin{equation}
\label{twirledA}
\mathcal T_{\operatorname{U}}(\mathcal A)=\sum_k \frac{1}{D^k}\beta_k\hat\Pi_k,
\end{equation}
in which  $\beta_k=\Tr(\mathcal A\hat\Pi_k^\dagger)$. This follows from linearity of integration and independence of the order of integration due to Fubini's theorem.
Indeed, we have:
\begin{eqnarray}
\mathcal T_{\operatorname{U}}(\mathcal A)&=&\mathcal T_{\operatorname{U}}\left(\int \left(A_{\textrm{n}}^{\otimes t}\right)^2
  d\mu(A)\right)=\int \mathcal T_{\operatorname{U}}\left(\left(A_{\textrm{n}}^{\otimes t}\right)^2\right)
  d\mu(A),\nonumber\\
  &=&\int\left[ \sum_k\frac{1}{D^k}\Tr\left(\left(A_{\textrm{n}}^{\otimes t}\right)^2\hat\Pi_k^\dagger\right)\hat\Pi_k\right]\,
  d\mu(A)\nonumber\\
  &=&\sum_k\frac{1}{D^k}\Tr\left(\left[\int\left(A_{\textrm{n}}^{\otimes t}\right)^2d\mu(A)\right]\hat\Pi_k^\dagger\right)\hat\Pi_k\,\nonumber\\
  &=&\sum_k\frac{1}{D^k}\Tr\left(\mathcal A\hat\Pi_k^\dagger\right)\hat\Pi_k.
\end{eqnarray}
Using  general form of the unitary twirling map  \eqref{fullUnitTwirlMap0} and the formula \eqref{twirledA}, we can finally simplify the map \eqref{mainTwirlSL5}:
\begin{eqnarray}
    &&\mathcal T_{\textrm{SL}}(\rho)=\mathcal T_{\operatorname{U}}(\rho)\mathcal T_{\operatorname{U}}(\mathcal A)\nonumber\\
    &&=\sum_{k}\frac{1}{D^k_L}\sum_{\lambda_1\lambda_2=1}^{D^k_V}\operatorname{Tr}\left(\rho\hat\Pi_{k}^{\lambda_1\lambda_2\dagger}\right)\hat\Pi_{k}^{\lambda_1\lambda_2}\sum_l \frac{1}{D^l}\beta_l\hat\Pi_l\nonumber\\
     &&=\sum_{k,l}\frac{1}{D^k_L}\left(\frac{\beta_l}{D^l}\right)\sum_{\lambda_1\lambda_2=1}^{D^k_V}\operatorname{Tr}\left(\rho\hat\Pi_{k}^{\lambda_1\lambda_2\dagger}\right)\hat\Pi_{k}^{\lambda_1\lambda_2}\hat\Pi_l\nonumber\\
     &&=\sum_{k,l}\frac{1}{D^k_L}\left(\frac{\beta_l}{D^l}\right)\sum_{\lambda_1\lambda_2=1}^{D^k_V}\operatorname{Tr}\left(\rho\hat\Pi_{k}^{\lambda_1\lambda_2\dagger}\right)\hat\Pi_{k}^{\lambda_1\lambda_2}\delta_{kl}\nonumber\\
     &&=\sum_{k}\frac{1}{D^k_L}\left(\frac{\beta_k}{D^k}\right)\sum_{\lambda_1\lambda_2=1}^{D^k_V}\operatorname{Tr}\left(\rho\hat\Pi_{k}^{\lambda_1\lambda_2\dagger}\right)\hat\Pi_{k}^{\lambda_1\lambda_2}.
    \label{mainTwirlSL6}
\end{eqnarray}
In the fourth line we utilised the property of operators $\hat \Pi_k^{\lambda_1\lambda_2}$ and $\hat\Pi_l$ shown in Appendix \ref{App:SWDuality}, formulas \eqref{PiIdent}.
If we express the closed form \eqref{fullUnitTwirlMap0} of the unitary twirling map in the compact form:
\begin{equation}
    \label{UnitTwirlSimp}
    \mathcal T_{\operatorname{U}}(\rho)=\sum_k\mathcal T_{\operatorname{U}}^{(k)}(\rho),
\end{equation}
in which:
\begin{equation}
    \label{UnitTwirlSimpK}
    T_{\operatorname{U}}^{(k)}(\rho)=\frac{1}{D^k_L}\sum_{\lambda_1\lambda_2=1}^{D^k_V}\operatorname{Tr}\left(\rho\hat\Pi_{k}^{\lambda_1\lambda_2\dagger}\right)\hat\Pi_{k}^{\lambda_1\lambda_2},
\end{equation}
the  SLOCC twirling map can be represented as:
\begin{equation}
    \label{SloccTwirlSimp}
    \mathcal T_{\textrm{SL}}(\rho)=\sum_k\left(\frac{\beta_k}{D^k}\right)\mathcal T_{\operatorname{U}}^{(k)}(\rho).
\end{equation}
It can be easily seen that all the impact of averaging over non-compact part of the $\textrm{SL}(d,\mathbb C)$ group is expressed in the rescaling of each of the component defined on an irreducible subspace by a factor $\left(\frac{\beta_k}{D^k}\right)$, in which $\beta_k=\Tr\left(\mathcal A\hat\Pi_k^\dagger\right)$ is a projection of an integral over the non-compact part of the SLOCC operation \eqref{Aop} onto the $k$-th irreducible subspace. This term directly depends on the integration measure  over the non-compact part $d\mu(A)$. An important difference with respect to the unitary case is that the SLOCC twirling map is not idempotent:
\begin{eqnarray}
\label{sloccnoidem}
 \mathcal T_{\textrm{SL}}( \mathcal T_{\textrm{SL}}(\rho))&=&\sum_k\frac{\beta_k}{D^k}\mathcal{T}_{\operatorname{U}}^{(k)}\left(\sum_l\frac{\beta_l}{D^l}\mathcal T_{\operatorname{U}}^{(l)}(\rho)\right)\nonumber\\
 &=&\sum_{kl}\frac{\beta_k\beta_l}{D^kD^l}\mathcal{T}_{\operatorname{U}}^{(k)}\left(\mathcal T_{\operatorname{U}}^{(l)}(\rho)\right)\nonumber\\
 &=&\sum_k\left(\frac{\beta_k}{D^k}\right)^2\mathcal T_{\operatorname{U}}^{(k)}(\rho).\nonumber\\
\end{eqnarray}
By induction a general rule holds:
\begin{equation}
    \label{sloccmapp}
     \mathcal{T}_{\textrm{SL}}^{m}(\rho)=\sum_k\left(\frac{\beta_k}{D^k}\right)^m\mathcal T_{\operatorname{U}}^{(k)}(\rho),
\end{equation}
in which:
\begin{equation}
\label{TSLm}
\mathcal T_{\textrm{SL}}^{m}=\underbrace{T_{\textrm{SL}}\circ\ldots\circ T_{\textrm{SL}}}_{m}.
\end{equation}
We can also easily calculate mean success probability $\langle p_M(\rho)\rangle$ of the SLOCC operations $M=KAK'$ over which we average, equal to the trace of $\mathcal T_{\textrm{SL}}(\rho)$ \eqref{pmav}:
\begin{eqnarray}
    \label{SloccTwirlSimpTr1}
    \langle p_M(\rho)\rangle= \Tr\left( \mathcal T_{\textrm{SL}}(\rho)\right)&=&\sum_k\left(\frac{\beta_k}{D^k}\right)\Tr\left(\mathcal T_{\operatorname{U}}^{(k)}(\rho)\right)\nonumber\\
    &=&\sum_k\left(\frac{\beta_k}{D^k}\right)\Tr\left(\rho\hat\Pi_k^\dagger\right).
\end{eqnarray}
This follows from the following transformations:
\begin{eqnarray}
&&\Tr\left(\mathcal T_{\operatorname{U}}^{(k)}(\rho)\right)=\Tr\left(\frac{1}{D^k_L}\sum_{\lambda_1\lambda_2=1}^{D^k_V}\operatorname{Tr}\left(\rho\hat\Pi_{k}^{\lambda_1\lambda_2\dagger}\right)\hat\Pi_{k}^{\lambda_1\lambda_2}\right)\nonumber\\
&&=\frac{1}{D^k_L}\sum_{\lambda_1\lambda_2=1}^{D^k_V}\operatorname{Tr}\left(\rho\hat\Pi_{k}^{\lambda_1\lambda_2\dagger}\right)\Tr\left(\hat\Pi_{k}^{\lambda_1\lambda_2}\right)\nonumber\\
&&=\frac{1}{D^k_L}\sum_{\lambda_1\lambda_2=1}^{D^k_V}\operatorname{Tr}\left(\rho\hat\Pi_{k}^{\lambda_1\lambda_2\dagger}\right)D^k_L\delta_{\lambda_1\lambda_2}\nonumber\\
&&=\sum_{\lambda=1}^{D^k_V}\operatorname{Tr}\left(\rho\hat\Pi_{k}^{\lambda\lambda\dagger}\right)=\operatorname{Tr}\left(\rho\sum_{\lambda=1}^{D^k_V}\hat\Pi_{k}^{\lambda\lambda\dagger}\right)=\operatorname{Tr}\left(\rho\hat\Pi_{k}^\dagger\right).\nonumber\\
\end{eqnarray}
We can summarise all presented derivations in the following theorem:
\begin{teor}
\label{maintheor}
The  collective SLOCC twirling map defined by formula \eqref{mainTwirlSL1} is represented by the following closed analytical form:
\begin{equation}
\label{finmapslteor}
    \mathcal T_{\operatorname{SL}}(\rho)=\sum_k\left(\frac{\beta_k}{D^k}\right)\mathcal T_{\operatorname{U}}^{(k)}(\rho),
\end{equation}
in which $\mathcal T_{\operatorname{U}}^{(k)}$ \eqref{UnitTwirlSimpK} is the unitary  twirling map acting on $k$-th irreducible subspace of the tensor space $(\mathbb C^d)^{\otimes t}$, $D^k$ is the dimension of the subspace, whereas the coefficient
\begin{equation}
\label{betafin}
\beta_k=\Tr\left(\left[\int\left(A_{\operatorname{n}}^{\otimes t}\right)^2d\mu(A)\right]\hat\Pi_k^\dagger\right)
\end{equation}
is a projection of an integral over the non-compact part of the SLOCC operation \eqref{Aop}  with respect to a normalised measure $d\mu(A)$ onto the $k$-th irreducible subspace. The average success probability of the SLOCC operations drawn according to assumed measure is specified by:
$$\langle p_M(\rho)\rangle=\sum_k\left(\frac{\beta_k}{D^k}\right)\Tr\left(\rho\hat\Pi_k^\dagger\right),$$
whereas $m$-times application of the map gives:
$$  \mathcal T_{\operatorname{SL}}^{m}(\rho)=\sum_k\left(\frac{\beta_k}{D^k}\right)^m\mathcal T_{\operatorname{U}}^{(k)}(\rho).$$
\end{teor}
The issue of \textit{noiseless subsystems} in the case of SLOCC collective twirling map \eqref{mainTwirlSL1} is subtle. Formally there are no states which are invariant with respect to this map in contrast to the unitary case, on the other hand the map \eqref{SloccTwirlSimp} does not change phase relations between the states within each of the irreducible subspaces. Therefore states of the form:
\begin{equation}
    \label{sloccdfsstate0}
    \rho_{\textrm{SL}}^{(k)}=\frac{1}{D^k_L}\sum_{\lambda_1\lambda_2=1}^{D^k_V}\rho^{k}_{\lambda_1\lambda_2}\hat\Pi_{k}^{\lambda_1\lambda_2},
\end{equation}
living within $k$-th subspace are untouched by the map \textit{upon postselection} to the $k$-th irreducible subspace.
Indeed, due to the formula \eqref{SloccTwirlSimp}, and the fact that the state \eqref{sloccdfsstate0} is invariant with respect to the unitary twirling we have:
\begin{equation}
    \label{sldfseffect}
    \mathcal T_{\operatorname{SL}}\left(\rho_{\textrm{SL}}^{(k)}\right)=\left(\frac{\beta_k}{D^k}\right)\rho_{\textrm{SL}}^{(k)}.
\end{equation}
What is modified in this case with respect to the unitary twirling is the projection probability $p_k$ to the subspace, which now reads:
\begin{equation}
    \label{pk}
    p_k=\left(\frac{\beta_k}{D^k}\right),
\end{equation}
and can be physically interpreted as the probability of success of postselection to the respective subspace. Note that due to the normalisation of the matrices $A_{\textrm n}$ we have the following relation:
\begin{equation}
\beta_k=\Tr\left(\left[\int\left(A_{\operatorname{n}}^{\otimes t}\right)^2d\mu(A)\right]\hat\Pi_k^\dagger\right)\leq \Tr\left(\hat \Pi_k\right)=D^k,
\end{equation}
hence $\beta_k\leq D^k$, which guarantees that the probability $p_k$ is well-defined. Further note that the choice of the integration measure $d\mu(A)$ determines the ordering of the coefficients $\{\beta_k\}$. The highest coefficient  $\beta_k$, which corresponds to  invariant subspaces $V^k_m$ with the dimension $V^k>1$ determines the noiseless subsystems with the highest projection probabilities $p_k$. These subsystems have the slowest rate of vanishing of the projection probability $p_k$ for the number $m$ of applications of the map $\mathcal T^m_{\textrm{SL}}$ \eqref{TSLm} going to infinity, $m\rightarrow\infty$.

In order to make all the above considerations less abstract let us discuss an example. Let us take the four-qubit system, in which according to our notation local Hilbert space dimension is $d=2$ and the tensor power is $t=4$. In this case the tensor space $(\mathbb C^2)^{\otimes 4}$ decomposes under collective action of the unitary group $\textrm{U}(2)^{\otimes 4}$ into three irreducible subspaces of dimensions respectively: $D^1=5$, $D^2=9$ and $D^3=2$. The first subspace of dimension $D^1=5$ is the subspace of fully symmetric states, and corresponds to a single (of multiplicity $D^1_V=1$) $5$-dimensional spin-$2$ irreducible representation of $\textrm{U}(2)$. The second subspace of dimension $D^2=9$ corresponds to three ($D^2_V=3$) equivalent $3$-dimensional spin-$1$ irreducible representations of $\textrm{U}(2)$, whereas finally the third subspace of dimension $D^3=2$ corresponds to two ($D^3_V=2)$ $1$-dimensional irreducible representations. The non-compact factor of the $\textrm{SL}(2,\mathbb C)$ group, which physically corresponds to a filtering operation, consists of matrices \eqref{ASL2C}, with norm equal to $\| A\|=x$, therefore a normalised filtering matrix in the formula \eqref{Aop} has the form:
\begin{equation}
\label{AnSL2C}
     A_{\operatorname{n}} =  \left( \begin{array}{cc} 1 & 0 \\ 0 & x^{-2}  \end{array}  \right),\,\, x \geq 1 .
\end{equation}
If we assume the following normalised integration measure on the non-compact component:
\begin{equation}
    \label{supmeasure}
    d\mu(A)=d\mu(x)=\frac{e^{-x}dx}{\int_1^{\infty}e^{-x}dx},
\end{equation}
the coefficients $\beta_k$ in \eqref{betafin} read respectively:
\begin{eqnarray}
\beta_1&\approx&0.30036,\nonumber\\
\beta_2&\approx&0.14652,\nonumber\\
\beta_3&\approx&0.12290.
\end{eqnarray}

\section{Duality of averaging for generalised Schur-Weyl-dual pairs}
\label{Sec:Conj}
Let us start by comparing closed formulas \eqref{finmapslteor} and \eqref{fullSymTwirlMap0} for twirling of quantum states over SLOCC operations and over permutations of subsystems written in a specific way:
\begin{eqnarray}
\label{dualpairSLSt}
\mathcal T_{\textrm{SL}}(\rho)&=&\sum_{k}\gamma^{\textrm{SL}}_k\sum_{\lambda_1\lambda_2=1}^{D^k_V}\operatorname{Tr}\left(\rho\hat\Pi_{k}^{\lambda_1\lambda_2\dagger}\right)\hat\Pi_{k}^{\lambda_1\lambda_2},\nonumber\\
 \mathcal T_{\textrm{sym}}(\rho)&=&\sum_{k}\gamma^{\textrm{sym}}_k\sum_{m_1,m_2=1}^{D^k_L}\operatorname{Tr}\left(\rho\hat\Pi^{m_1 m_2\dagger}_k\right)\hat\Pi^{m_1 m_2}_k,\nonumber\\
\end{eqnarray}
in which the coefficients read:
\begin{equation}
    \gamma^{\textrm{SL}}_k=\frac{1}{D^k_L}\left(\frac{\beta_k}{D^k}\right),\,\,\,\gamma^{\textrm{sym}}_k=\frac{1}{D^k_V}.
\end{equation}
These two operations correspond to averaging of quantum states with respect to elements of a dual reductive pair: $\{\textrm{SL}(d,\mathbb C)^{\otimes t}, \textrm{S}_t\}$ consisting of the special linear group acting collectively on $t$ subsystems and the symmetric group of $t$-element set. Comparing the two formulas it can be seen that up to the coefficients these two operations project quantum states (understood as density operators) into operator subspaces  the action of which leaves invariant subspaces irreducible with respect to the second counterpart of the pair. In this section we show, that this result is not specific to the already discussed pairs $\{\textrm{SL}(d,\mathbb C)^{\otimes t}, \textrm{S}_t\}$ and $\{\textrm{U}(d)^{\otimes t}, \textrm{S}_t\}$, but holds for arbitrary dual reductive pair. Our proof will be based on showing, that crucial derivations already presented do not depend on the specific choice of the group under consideration, but hold due to Schur-Weyl duality for the corresponding groups.

Let us start by introducing the idea of a \textit{dual reductive pair} in a more formal way. The crucial concept is that of a \textit{reductive group}, which is any algebraic subgroup of a general linear group, the \textit{rational} \footnote{By rational representation one means a matrix representation of a matrix Lie group in which matrix elements of the representation are rational functions of matrix elements of the represented group (see \cite{Goodman09}, section 1.5).} representations of which are \textit{completely reducible}, which means that they are direct sums of irreducible representations \cite{Goodman09}. 
\begin{deff}
\label{dualpair}
Let $\operatorname{G}_1$ and $\operatorname{G}_2$ be two reductive subgroups of the general linear group $\operatorname{GL}(n,\mathbb C)$. A pair of groups $\{\operatorname{G}_1,\operatorname{G}_2\}$ is called a dual reductive pair if and only if matrix algebras generated by elements of these groups are mutual commutants in the full matrix algebra $\operatorname{End}(\mathbb C^n).$
\end{deff}
For any dual reductive pair a Schur-Weyl duality holds (see \cite{Goodman09}, Prop. 9.2.1), which implies that there exists a vector basis on $\mathbb C^n$, called by us \textit{Schur vector basis}, which via outer product of basis vectors defines two operator bases, called by us \textit{Schur operator bases}. These two operator bases block-diagonalise actions of operators from both groups, in full analogy to the Schur operator bases \eqref{PiBasis0} introduced already in Section \ref{Sec:SWIntro} for the dual pair $\{\textrm{U}(d)^{\otimes t}, \textrm S_t\}$. For simplicity let us keep the entire notation from that section regarding Schur vector basis and Schur operator bases, with the convention, that the subspaces $L^i_{\lambda}$ in \eqref{SchurBasis0} are irreducible under the action of $\operatorname{G}_1$, whereas the subspaces $V^i_m$ are irreducible under the action of $\operatorname{G}_2$. All properties of the corresponding Schur  operator bases \eqref{PiBasis0}, which follow from orthogonality of Schur vector basis naturally apply in the current generalised context. Then arbitrary matrix representants $G_1,G_2$ of the groups $\operatorname{G}_1,\operatorname{G}_2$ have the following Schur operator bases decompositions:
\begin{eqnarray}
\label{dualpairG1G2}
G_1&=&\sum_{k}\frac{1}{d^k_L}\sum_{\lambda_1\lambda_2=1}^{d^k_V}\operatorname{Tr}\left(G_1\hat\Pi_{k}^{\lambda_1\lambda_2\dagger}\right)\hat\Pi_{k}^{\lambda_1\lambda_2},\nonumber\\
 G_2&=&\sum_{k}\frac{1}{d^k_V}\sum_{m_1,m_2=1}^{d^k_L}\operatorname{Tr}\left(G_2\hat\Pi^{m_1 m_2\dagger}_k\right)\hat\Pi^{m_1 m_2}_k.\nonumber\\
\end{eqnarray}
Note that above we use a modified notation for the dimensions of irreducible subspaces (small letters $d^k$) in order not to confuse them with dimensions for irreducible subspaces for collective unitary actions and permutations specified by formulas \eqref{WCHF} and \eqref{HLF}. Note that the matrix algebras spanned  by the operators of the type of $G_1$ and $G_2$ are closed under multiplication and Hermitian conjugate, therefore they form \textit{finite dimensional von Neumann algebras} (in the context of Schur-Weyl duality for compact groups this fact has been stated in \cite{Marvian14}). We will utilise this important fact in the proof of the main result of this section.

The crucial issue in our generalisation is that any reductive Lie group $\textrm G$ has a Cartan `KAK' decomposition (see \cite{Knapp}, section VII.3), which essentially means that any element $G\in \textrm G$ can be represented as $G=KAK'$, in which $K,K'$ are elements of a maximal compact subgroup $\textrm K$ of $\textrm G$, whereas $A$ is an element of a maximal abelian subgroup $\textrm A$ of $\textrm G$. We can therefore define a \textit{general stochastic operation} on a state space $\mathbb C^n$ of $n$-level quantum systems via:
\begin{equation}
    \label{gsodef}
    \mathcal G(\rho)=\left(KA_{\textrm{n}}K'\right)\rho \left(KA_{\textrm{n}}K'\right)^{\dagger},
\end{equation}
in which $K,K'$ are matrices from a unitary representation of the maximal compact subgroup $\textrm K$ on $\mathbb C^n$ and $A_{\textrm n}=A/\|A\|$ is a normalised matrix belonging to a representation of the maximal abelian subgroup $\textrm A$ on $\mathbb C^n$. Some comments are needed here. Firstly, in full analogy with the SLOCC map \eqref{SLOCCMap}, the map defined by \eqref{gsodef} is trace non-increasing (for a formal proof see remark in the Appendix \ref{App:LemTrace}) and preserves positivity, therefore can be interpreted as a valid quantum operation. Secondly, in the above definition we do not assume any internal structure of the state space $\mathbb C^n$, so the issue whether it is a single-system space or it represents several subsystems is left free for concrete implementations. Therefore the map \eqref{gsodef} should be thought of as an abstract interface for defining concrete quantum maps with concrete physical interpretation. In the same manner we define an abstract generalisation of the twirling map to the case of $\textrm G$-twirling with respect to any reductive group $\textrm G$ by the iterated integral with respect to the Cartan decomposition, in full  analogy to \eqref{mainTwirlSL2}:
\begin{equation}
    \mathcal T_{\textrm{G}}(\rho)=\int \left(KA_{\textrm{n}}K'\right)\rho \left(KA_{\textrm{n}}K'\right)^{\dagger} d\mu(K)d\mu(A)d\mu(K'),
    \label{mainTwirlG0}
\end{equation}
in which we assume, that $d\mu(K)$ is a normalised left- and right-invariant Haar measure on the compact group $\textrm K$ (which in the case of $\textrm K$ being a discrete group is simply a \textit{counting measure}), whereas $d\mu(A)$ is any normalised measure on the group manifold of $\textrm A$. The map \eqref{mainTwirlG0} should be treated as an abstract interface for defining any meaningful twirling maps on quantum states rather than as a concrete physical operation. Note that all previously discussed twirling maps can be represented in that way.

Now we are ready to state the main result of this section:
\begin{teor}
\label{teorG1G1}
Let us take any dual reductive pair $\{\operatorname{G}_1,\operatorname{G}_2\}$ acting naturally on the complex vector space $\mathbb C^n$. Let us assume the notation for Schur basis decomposition of arbitrary elements $G_1\in \operatorname{G}_1$ and $G_2\in\operatorname{G}_2$ as presented in formulas \eqref{dualpairG1G2}. Let us also assume Cartan decompositions of these elements as $G_1=K^{(1)}A^{(1)}K^{(1)'}$ and $G_2=K^{(2)}A^{(2)}K^{(2)'}$. 
Then the generalised twirling maps with respect to these groups defined via iterated integral over Cartan decomposition of these groups \eqref{mainTwirlG0} are in the following dual relation:
\begin{eqnarray}
\label{GtwirlMaps}
\mathcal T_{\operatorname{G}_1}(\rho)&=&\sum_{k}\frac{1}{d^k_V}\left(\frac{\beta^{\operatorname{G}_1}_k}{d^k}\right)\sum_{m_1m_2=1}^{d^k_L}\operatorname{Tr}\left(\rho\hat\Pi_{k}^{m_1m_2\dagger}\right)\hat\Pi_{k}^{m_1m_2},\nonumber\\
\mathcal T_{\operatorname{G}_2}(\rho)&=&\sum_{k}\frac{1}{d^k_L}\left(\frac{\beta^{\operatorname{G}_2}_k}{d^k}\right)\sum_{\lambda_1\lambda_2=1}^{d^k_V}\operatorname{Tr}\left(\rho\hat\Pi_{k}^{\lambda_1\lambda_2\dagger}\right)\hat\Pi_{k}^{\lambda_1\lambda_2}.\nonumber\\
\end{eqnarray}
The coefficients $\beta_k^{\operatorname{G_i}}$ are in general expressed as integrals over the noncompact factors of the Cartan decomposition of both groups:
\begin{eqnarray}
\label{betafinG}
\beta_k^{\operatorname{G}_i}=\Tr\left(\left[\int A_{\operatorname{n}}^{(i)}A_{\operatorname{n}}^{(i)\dagger}d\mu_i\left(A^{(i)}\right)\right]\hat\Pi_k^\dagger\right),
\end{eqnarray}
and reduce to $\beta_k^{\operatorname{G}_i}=d^k$ for compact group $\operatorname{G}_i$.
\end{teor}
\begin{proof}
The proof is actually a generalisation of the proof for the SLOCC twirling map presented in the previous section. For convenience we will proceed with the proof for the twirling map $\mathcal T_{\operatorname{G}_1}$, the proof for the map $\mathcal T_{\operatorname{G}_2}$ is entirely analogous so will be skipped. Firstly note that the twirling map $\mathcal T_{\operatorname{K}_1}$ with respect to the maximally compact subgroup $\operatorname{K_1}$ of $\operatorname{G}_1$ has entirely the same form as the unitary twirling map with respect to the entire unitary group \eqref{fullUnitTwirlMap0}, namely it reads:
\begin{equation}
\mathcal T_{\operatorname{K}_1}(\rho)=\sum_{k}\frac{1}{d^k_V}\sum_{m_1m_2=1}^{d^k_L}\operatorname{Tr}\left(\rho\hat\Pi_{k}^{m_1m_2\dagger}\right)\hat\Pi_{k}^{m_1m_2}.\nonumber\\
\end{equation}
This is because the Haar measure on any compact group is left and right invariant, therefore one can perform the entire derivation for the unitary twirling map without any modification (for alternative proof of general twirling map with respect to any compact symmetry group see \cite{Bartlett07}, section II.C). Secondly note that both elements $K^{(1)}$ and $A^{(1)}$ of the Cartan decomposition of any matrix $G_1$ have  decompositions analogous to $G_1$ in terms of  the corresponding Schur operator basis \eqref{dualpairG1G2}, namely:
\begin{eqnarray}
\label{decompK1A1}
K^{(1)}&=&\sum_{k}\frac{1}{d^k_L}\sum_{\lambda_1\lambda_2=1}^{d^k_V}\operatorname{Tr}\left(K^{(1)}\hat\Pi_{k}^{\lambda_1\lambda_2\dagger}\right)\hat\Pi_{k}^{\lambda_1\lambda_2},\nonumber\\
 A^{(1)}&=&\sum_{k}\frac{1}{d^k_L}\sum_{\lambda_1\lambda_2=1}^{d^k_V}\operatorname{Tr}\left(A^{(1)}\hat\Pi_{k}^{\lambda_1\lambda_2\dagger}\right)\hat\Pi_{k}^{\lambda_1\lambda_2}.\nonumber\\
\end{eqnarray}
Since operators $\hat\Pi_{k}^{\lambda_1\lambda_2}$ and $\hat\Pi_{k}^{m_1m_2}$ commute for any Schur operator bases (see \eqref{PiCommute}), the generalisation of Lemma  \ref{lem1} holds and the operator $\mathcal T_{\operatorname{K}_1}(\rho)$ commutes with both $K^{(1)}$ and $A^{(1)}$. Therefore the generalised twirling map \eqref{mainTwirlG0} factorises in analogy to the SLOCC twirling map:
\begin{eqnarray}
    &&\mathcal T_{\textrm{G}_1}(\rho)=\int \left(K^{(1)}A^{(1)}_{\textrm{n}}K'^{(1)}\right)\rho \left(K^{(1)}A^{(1)}_{\textrm{n}}K'^{(1)}\right)^{\dagger}\nonumber\\ &&d\mu_1\left(K^{(1)}\right)d\mu_1\left(A^{(1)}\right)d\mu_1\left(K'^{(1)}\right)\nonumber\\
    &&=\int K^{(1)}A^{(1)}_{\textrm{n}}\left(\int K'^{(1)}\rho K'^{(1)\dagger}d\mu_1\left(K'^{(1)}\right)\right)A^{(1)\dagger}_{\textrm{n}}K^{(1)\dagger}\nonumber\\ &&d\mu_1\left(K^{(1)}\right)d\mu_1\left(A^{(1)}\right)\nonumber\\
     &&=\int K^{(1)}A^{(1)}_{\textrm{n}}\mathcal T_{\operatorname{K}_1}(\rho)A^{(1)\dagger}_{\textrm{n}}K^{(1)\dagger}d\mu_1\left(K^{(1)}\right)d\mu_1\left(A^{(1)}\right)\nonumber\\
       &&=T_{\operatorname{K}_1}(\rho)\int K^{(1)}A^{(1)}_{\textrm{n}} A^{(1)\dagger}_{\textrm{n}}K^{(1)\dagger}d\mu_1\left(K^{(1)}\right)d\mu_1\left(A^{(1)}\right)\nonumber\\
         &&=T_{\operatorname{K}_1}(\rho)T_{\operatorname{K}_1}\left(\int A^{(1)}_{\textrm{n}} A^{(1)\dagger}_{\textrm{n}}d\mu_1\left(A^{(1)}\right)\right).
    \label{mainTwirlG}
\end{eqnarray}
We have obtained a factorisation analogous to the case of twirling with respect to the SLOCC operations. Now the crucial point is to show that an analogy of Lemma \ref{lem2} holds also in the general case:
\begin{lem}
\label{lem2gen}
Let $G_2$ be arbitrary element of matrix von Neumann algebra spanned by representation matrices of the group $\operatorname{G}_2$. Then the operator $\mathcal T_{\operatorname{K}_1}(G_2)$ is fully diagonal in the \textit{Schur operator basis} related with dual reductive pair $\{\operatorname{G}_1, \operatorname{G}_2\}$:
\begin{equation}
\label{twirledProdOpgen}
\mathcal T_{\operatorname{K}_1}(G_2)=\sum_k \frac{1}{d^k}\alpha_k\hat\Pi_k,  
\end{equation}
in which the coefficients read: $\alpha_k=\Tr(G_2\hat\Pi_k^\dagger)$, $\hat\Pi_k$ is a projector onto the $i$-th irreducible subspace \eqref{iProj} and $d^k$ is the dimension of the subspace: $d^k=d^k_Ld^k_V$.
\end{lem}
The proof of the above lemma is entirely analogous to the one presented in Appendix \ref{proof:lem2}, and holds due to the fact that the mentioned proof is based solely on orthogonality properties of the Schur basis, which do not depend on the structure of groups under consideration. 

Now, since as mentioned before, the algebra generated by operators $G_1$ is a von Neumann algebra, therefore the operator $\int A^{(1)}_{\textrm{n}} A^{(1)\dagger}_{\textrm{n}}d\mu\left(A^{(1)}\right)$ also belongs to this algebra and consequently Lemma \ref{lem2gen} implies, that:
\begin{equation}
\label{twirledProdOpgen1}
\mathcal T_{\operatorname{K}_1}\left(\int A^{(1)}_{\textrm{n}} A^{(1)\dagger}_{\textrm{n}}d\mu_1\left(A^{(1)}\right)\right)=\sum_k \frac{\beta_k^{\operatorname{G}_1}}{d^k}\hat\Pi_k,  
\end{equation}
in which the coefficient $\beta_k^{\operatorname{G}_1}$ is specified by formula \eqref{betafinG} for $i=1$. Derivation of the final form of the generalised twirling map $\mathcal T_{\operatorname{G}_1}$ is now entirely analogous to the derivation performed in a sequence of formulas \eqref{mainTwirlSL6}, which completes the proof.
\end{proof}
The entire issue of \textit{noiseless subsystems} can be also restated in the general case of twirling with respect to dual reductive pairs. Namely states of the form:
\begin{eqnarray}
    \label{sloccdfsstate}
    \rho_{\operatorname{G_1}}^{(k)}&=&\frac{1}{d^k_V}\sum_{m_1m_2=1}^{d^k_L}\rho^{k}_{m_1m_2}\hat\Pi_{k}^{m_1m_2}\nonumber\\
      \rho_{\operatorname{G_2}}^{(k)}&=&\frac{1}{d^k_L}\sum_{\lambda_1\lambda_2=1}^{d^k_V}\rho^{k}_{\lambda_1\lambda_2}\hat\Pi_{k}^{\lambda_1\lambda_2}
\end{eqnarray}
are invariant with respect to twirling operations \eqref{GtwirlMaps} on condition that a postselection to the $k$-th subspace is performed, which succeeds with probability respectively $\frac{\beta^{\operatorname{G}_1}_k}{d^k}$ and $\frac{\beta^{\operatorname{G}_2}_k}{d^k}$. If any of the groups has finite dimensional unitary representation (which holds if the group is a compact Lie group or a finite group), then $\beta_k=d^k$, and therefore the respective noiseless subsystem is unconditional.

\section{Conclusions}
\label{Sec:Conc}
In this work we introduced a very general averaging operation on the state space of any finite dimensional quantum system with respect to an arbitrary symmetry group possessing the Cartan `KAK' decomposition. The averaging procedure is defined via iterated integral over the compact and non-compact components of the Cartan decomposition of the symmetry group. This very general notion of averaging of quantum states allows us to significantly generalise known results concerning averaging of quantum states over collective unitary operations. 

In a typical scenario of averaging multipartite quantum states over collective unitary operations the effect of such averaging is determined by a mathematical concept called Schur-Weyl duality, which introduces the notion of a dual reductive pair, a pair of two matrix groups that as matrix algebras are mutual commutants. It turns out that the dual counterpart to the collective action of the unitary group is the symmetric group. Then Schur's Lemmas imply, that averaging over one of the elements of the dual pair just projects the quantum state onto a subspace spanned by elements of the second counterpart.
We called this relation \textit{duality of averaging}. Our main result in this work is showing, that such defined duality of averaging persists in the case of arbitrary two symmetry groups, which are dual reductive pairs. We have shown this fact in three steps:
\begin{itemize}
    \item firstly we analysed the case of averaging multipartite quantum states over Cartan-decomposed collective SLOCC operations represented by the tensor power of the non-compact symmetry group $\operatorname{SL}(d,\mathbb C)$, which also forms a dual reductive pair with the symmetric group; as a result we obtained a factorisation of such an averaging map into two components, the first being a collective unitary averaging map applied to the input state, whereas the second being unitary averaging map applied to the integral over  non-compact components of the SLOCC operations;
    \item secondly we noticed that such a factorisation implies that the result of averaging quantum states over SLOCC operations has the same block-diagonal structure as the unitary averaging, but endowed with additional correcting coefficients;
    this means that averaging of quantum states with respect to SLOCC operations and with respect to permutation of subsystems also conforms to the duality of averaging;
    \item finally we have shown that the proof for the SLOCC case directly generalises to the averaging procedure over any two groups forming a dual reductive pair, therefore the duality of averaging holds for averaging over any dual pair of symmetry groups.
\end{itemize}
The above points summarise the structural aspect of our work, which shows in a very general context the form of averaging maps for arbitrary finite dimensional quantum systems.

The main physical consequence of our result appears in the context of the so called \textit{noiseless subsystems}. It is a well-known consequence of Schur-Weyl duality, that quantum states spanned by the operator basis of permutation operators are untouched by the collective unitary  averaging. If we interpret collective unitary operations as a collective quantum noise, than such states are immune to this type of noise.
Noiseless subsystems can be also used to encode quantum information in the case of lack of a common reference frame between the sender and receiver of a given state. Our result shows, that the idea of noiseless subsystems persists to the case of averaging over arbitrary symmetry groups forming dual reductive pairs \textit{in a conditional way}: if one of the groups in the pair is non-compact, then noiseless subsystems persist on condition that we make a postselection to the subsystem: all the input from the non-unitary part of the averaging process is the effect of filtering, which means that in some runs of the experiment a system will not be found in a given subspace. \textit{However all the block-diagonal structure of noiseless subsystems is solely determined by the unitary part of the averaging process.}

All the discussion within this work concerns the case of finite dimensional quantum systems. In our opinion an interesting open problem would be to search for analogues of duality of averaging in the context of averaging infinite dimensional quantum systems over unitary representations of symmetry groups. A natural example would be averaging quantum states of massive and massless particles over unitary representations of Poincare group \cite{Dragan15} or Galilean group \cite{Piani16}.

\section{Acknowledgements}

We acknowledge partial support by the Foundation for Polish Science (IRAP project, ICTQT, contract no. MAB/2018/5, co-financed by EU within Smart Growth Operational Programme).

\bibliography{TDesigns}
\newpage
\appendix 

\begin{widetext}
\section{Direct analytical form of the Unitary and Symmetric Twirling maps}
\label{App:UnitDFS}

\subsection{Construction of Schur vector basis and Schur operator basis for  the dual pair $\{\textrm U(d)^{\otimes t},\textrm{S}_t\}$.}
\label{App:SWDuality}

In this section we provide a detailed  construction of a Schur vector basis and Schur operator basis for the joint action of the product group $\textrm U(d)^{\otimes t}\times\textrm{S}_t$ on the tensor space $(\mathbb C^d)^{\otimes t}$. The presentation below is inspired by the book of W.-K. Tung \cite{Tung}, an alternative construction using cascaded Clebsch-Gordan decompositions can be found in the PhD thesis of A. Harrow \cite{HarrowPHD}.

Distinct irreducible representations of the action of $\textrm U(d)^{\otimes t}\times\textrm{S}_t$ on $(\mathbb C^d)^{\otimes t}$ are labelled by distinct Young diagrams with $t$ boxes and \textit{at most} $d$ rows. These representations can be uniquely provided by construction of a Schur vector basis, which explicitly encodes the subspaces of $(\mathbb C^d)^{\otimes t}$ which are irreducible with respect to the collective action of the unitary group and of the action of the symmetric group.

Before proceeding with construction of  Schur bases let us briefly describe how to construct matrix irreducible representations of the symmetric group related with a given Young diagram with the numbers of boxes in subsequent rows given by $r_1,\ldots,r_k$, $k\leq d$ and $\sum_j r_j=t$. Firstly let us introduce the notion of \textit{standard Young tableau}, which is a $t$-box Young diagram filled in with integers $1,\ldots,t$ in such a way that the numbers are increasing row-wise from left to right and column-wise from top to bottom, but not necesserilly in a strict order. With a given standard Young tableau we associate so-called \textit{Young symmetrizer}, which is a building block for constructing irreducible representations of the symmetric group. Abstractly, a Young symmetrizer is an element of a \textit{group algebra} of the symmetric group, that is a vector space involving formal linear combinations of group elements endowed with natural multiplication inherited from the group composition law. A Young symmetrizer is defined as:
\begin{equation}
    \label{YoungSym}
    Y_{r_1,\ldots,r_k}=\frac{1}{kl}\sum_{r=1}^k\sum_{c=1}^l \textrm{sgn}(p_c) p_r p_c,
\end{equation}
where the sum runs over all permutations preserving all rows and columns of the standard Young tableau, weighted with the parity of the column permutations. A matrix representation of a Young symmetrizer is simply specified by:
\begin{equation}
    \label{YoungSymRep}
    \hat{Y}_{r_1,\ldots,r_k}=\frac{1}{kl}\sum_{r=1}^k\sum_{c=1}^l \textrm{sgn}(p_c) O_{p_r} O_{p_c},
\end{equation}
in which the orthogonal matrices $O_p$ are already defined in formula \eqref{ORep} in the main text.

Let us label all inequivalent irreducible representations of $\textrm{U}(d)^{\otimes t}$ and $\textrm S_t$ corresponding to a given Young diagram by an index $i$.
Although the index $i$ labels uniquely all \textit{inequivalent} representations of both groups, it may happen that there exist \textit{equivalent} representations of both of them, which however correspond to \textit{different subspaces} of  $(\mathbb C^d)^{\otimes t}$. In order to describe this phenomenon let us construct a Schur basis for $i$-th irreducible subspace:
\begin{itemize}
\item let us denote a matrix representation of a Young symmetrizer related with a fixed irreducible subspace $i$ and of a fixed allowed filling with integers $\lambda$ as $\hat Y_i^{\lambda}$; such Young symmetrizer corresponds to a fixed standard Young tableau with a shape uniquely specified by lengths of its rows $r_1,\ldots,r_k$, therefore the index $i$ stays in one-to-one correspondence with a shape of the diagram;
    \item for every $\lambda$-th standard Young tableau of the $i$-th type, consisting of boxes organized into $k$ rows of lengths $r_{\mu}$, $1\leq \mu \leq k$, let us take the corresponding Young symmetrizer $\hat Y_i^{\lambda}$, and let us construct a linear subspace $L^i_{\lambda}=\{\hat Y_i^{\lambda}\ket{v}\},\,\ket{v}\in(\mathbb C^d)^{\otimes t}$, consisting of tensors of the \textit{symmetry type} $Y_i^{\lambda}$. The dimension of each space $L^i_{\lambda}$ is specified by the so-called Weyl character formula:
    \begin{equation}
        \label{WCHF}
        \textrm{dim}(L^i_{\lambda})=D_L^i=\prod_{1\leq \mu<\nu\leq k}\frac{r_{\mu}-r_{\nu}+\nu-\mu}{\nu-\mu},
    \end{equation}
    in which the product runs over all rows in the Young diagram of type $i$;
    \item note that each subspace $L_{\lambda}^i$ is invariant under the action of $U^{\otimes t}$, since: $\{U^{\otimes t}\hat Y_i^{\lambda}\ket{v}\}=\{\hat Y_i^{\lambda}U^{\otimes t}\ket{v}\}=\{\hat Y_i^{\lambda}\ket{v}\}$;
    \item let us assume that $\lambda=1$ corresponds to a \textit{normal Young tableau}, that is a tableau with natural row-wise filling of integers from left to right. Let us denote an orthonormal basis  of the subspace $L^i_1=\{\hat Y_i^1\ket{v}\}$ as $\{\ket{i,m,1}\}_{m=1}^{D_L^i}$. Since all other Young symmetrizers $\hat Y_i^{\lambda}$ from $i$-th class can be obtained from $\hat Y_i^1$ by action of some permutation $p_{\lambda}$, we have $\hat Y_i^{\lambda}=O_{p_{\lambda}}\hat Y_i^{1}$, and therefore orthonormal bases for other subspaces $L^i_{\lambda}$ can be obtained from the basis of $L_1^i$: 
    \begin{equation}
    \{\ket{i,m,\lambda}\}_{m=1}^{D_L^i}=     \{O_{p_{\lambda}}\ket{i,m,1}\}_{m=1}^{D_L^i};
    \end{equation}
    \item in this way we have obtained orthonormal bases $\{\ket{i,m,\lambda}\}_{m=1}^{D_L^i}$  for the subspaces $L^i_{\lambda}$; 
    \item on the other hand the subspaces $V^i_{m}$ generated by sets of vectors $\{\ket{i,m,\lambda}\}_{\lambda=1}^{D_V^i}$ are invariant under the action of the symmetric group $\textrm{S}_t$; they are called \textit{Specht moduli}; the dimension of $V^i_{m}$ is equal to the number of inequivalent \textit{standard Young tableau's} of a fixed shape and is specified by the so-called \textit{hook-length formula}:
     \begin{equation}
        \label{HLF}
        \textrm{dim}(V^i_{m})=D_V^i=\frac{t!}{\prod_{\mu,\nu}h(\mu,\nu)},
    \end{equation}
    in which the \textit{hook-length} $h(\mu,\nu)$ corresponding to a box placed in $\mu$-th row and $\nu$-th column is defined as the number of boxes in a hook constructed by taking all the boxes to the right and to the bottom; the product is performed over hook lengths corresponding to all boxes in a tableau;
    \item although the vectors $\{\ket{i,m,\lambda}\}_{\lambda=1}^{D_V^i}$ spanning the spaces $V^i_m$ are typically non-orthogonal, they can be orthogonalised, in such a way that each new vector belongs to the same symmetry type, and therefore to the same subspace $L^i_{\lambda}$; in this way we obtain an orthonormal basis for each of the symmetry type $i$, which can be presented as an array:
    \begin{eqnarray}
    \label{SchurBasis}
     L^i_1:\,\,&&\ket{i,1,1}\,\,\,\,\,\,\,\,\,\ket{i,2,1}\,\,\,\,\,\,\ldots\,\,\,\,\,\,\ket{i,D^i_L,1} \nonumber\\
     L^i_2:\,\,&&\ket{i,1,2}\,\,\,\,\,\,\,\,\,\ket{i,2,2}\,\,\,\,\,\,\ldots\,\,\,\,\,\,\ket{i,D^i_L,2} \nonumber\\
     && \ldots\ldots\ldots\ldots\ldots\ldots\ldots\ldots\ldots\ldots\ldots\ldots \nonumber\\
       L^i_{D^i_V}:\,\,&&\underbrace{\ket{i,1,D^i_V}}_{V^i_1}\,\,\underbrace{\ket{i,2,D^i_V}}_{V^i_2}\,\,\,\ldots\,\,\,\underbrace{\ket{i,D^i_L,D^i_V}}_{V^i_{D^i_L}}\nonumber\\
    \end{eqnarray}
    \item the above arrangement of the Schur vector basis indicates the essence of \textit{Schur-Weyl duality}: the basis vectors $\ket{i,m,\lambda}$ within each irreducible subspace $i$ can be seen as belonging to $D^i_V$ subspaces $L^i_{\lambda}$, which correspond to equivalent irreducible $D^i_L$-dimensional representations of the unitary group $\textrm{U}(d)$, or as belonging to $D^i_L$ subspaces $V^i_m$, which correspond to equivalent irreducible $D^i_V$-dimensional representations of the symmetric group $\textrm{S}_t$; 
\end{itemize}
After introducing the construction of the Schur vector basis, lets come back to the block-diagonalisation of the action of unitary and symmetric group on the tensor product space $(\mathbb C^d)^{\otimes t}$. For this aim we introduce the following operator basis for the space $\textrm{End}(\mathbb C^d)^{\otimes t}$:
\begin{equation}
    \label{FullPiBasis}
    \hat\Pi_{ij}^{m_1\lambda_1 m_2\lambda_2}=\ket{i,m_1,\lambda_1}\bra{j,m_2,\lambda_2}.
\end{equation}
The above operators fulfill the following property, which comes from orthonormality of the Schur vector basis \eqref{SchurBasis}:
\begin{equation}
    \label{SBOrt}
\hat\Pi_{ij}^{m_1\lambda_1 m_2\lambda_2}\hat\Pi_{kl}^{n_1\mu_1 n_2\mu_2}=\delta_{jk}\delta_{\lambda_2\mu_1}\delta_{m_2n_1}\hat\Pi_{il}^{m_1\lambda_1 n_2\mu_2}.    
\end{equation}
Note that due to the fact that operators \eqref{FullPiBasis} are outer products of basis vectors, their Hermitian conjugate is equivalent to transposition of corresponding indices:
\begin{equation}
    \label{FullPiBasisCT}
     \hat\Pi_{ij}^{m_1\lambda_1 m_2\lambda_2\dagger}=\ket{j,m_2,\lambda_2}\bra{i,m_1,\lambda_1}=\hat\Pi_{ji}^{m_2\lambda_2 m_1\lambda_1}.
\end{equation}
Let us also introduce two other operators:
\begin{eqnarray}
\label{PiBasis}
\hat\Pi^{\lambda_1\lambda_2}_i&=&\sum_{m=1}^{D^i_L}\hat\Pi_{ii}^{m\lambda_1 m\lambda_2},\nonumber\\
\hat\Pi^{m_1 m_2}_i&=&\sum_{\lambda=1}^{D^i_V}\hat\Pi_{ii}^{m_1\lambda m_2\lambda}.
\end{eqnarray}
We refer to the sets of operators $\{\hat\Pi_i^{\lambda_1\lambda_2}\}_{\lambda_1\lambda_2}$ and $\{\hat\Pi_i^{m_1m_2}\}_{m_1m_2}$ as \textit{Schur operator bases}.
They are especially important, since they represent the block-diagonalisation of the action of both groups on the tensor space, namely:
\begin{equation}
    \label{Ut}
    U^{\otimes t}=\sum_{i}\frac{1}{D^i_V}\sum_{m_1,m_2=1}^{D^i_L}\operatorname{Tr}\left(U^{\otimes t}\hat\Pi^{m_1 m_2\dagger}_i\right)\hat\Pi^{m_1 m_2}_i,
\end{equation}
and analogously:
\begin{equation}
    \label{Op}
    O_{p\in S_t}=\sum_{i}\frac{1}{D^i_L}\sum_{\lambda_1,\lambda_2=1}^{D^i_V}\operatorname{Tr}\left(O_p\hat\Pi^{\lambda_1 \lambda_2\dagger}_i\right)\hat\Pi^{\lambda_1 \lambda_2}_i.
\end{equation}
The normalisation factors come from the fact, that:
\begin{eqnarray}
\label{PiBasisTrace}
\Tr(\hat\Pi^{m_1 m_2}_i)&=&\delta_{m_1m_2}D^i_V,\nonumber\\
\Tr(\hat\Pi^{\lambda_1 \lambda_2}_i)&=&\delta_{\lambda_1 \lambda_2}D^i_L.
\end{eqnarray}
Both the  formulas \eqref{Ut} and \eqref{Op} are simple consequences of the construction of the Schur vector basis \eqref{SchurBasis}, namely that the subspaces $L^i_{\lambda}$ correspond to equivalent irreducible representations of the unitary group, and that the subspaces $V^i_m$ correspond to equivalent irreducible representations of the symmetric group. In order to \textit{see} the block-diagonalisation in a matrix representation of the above operators, let us introduce two matrix bases in the space of operators $\textrm{End}(\mathbb C^d)^{\otimes t}$:
\begin{itemize}
    \item \textbf{U-basis}: for each irreducible subspace $i$ we order the vectors in the table \eqref{SchurBasis} \textit{row-wise} and then enumerate all such ordered vectors for all subspaces $i$ by a single index $\mu$; let us denote such obtained basis by $\{\ket{e_{\mu}}\}_{\mu=1}^{d^t}$; then arbitrary operator $U^{\otimes t}$ is block-diagonal in matrix basis $\{\ket{e_{\mu}}\bra{e_{\nu}}\}_{\mu,\nu=1}^{d^t}$;
    \item \textbf{S-basis}:  for each irreducible subspace $i$ we order the vectors in the table \eqref{SchurBasis} \textit{column-wise} and then enumerate all such ordered vectors for all subspaces $i$ by a single index $\mu$; let us denote such obtained basis by $\{\ket{\tilde e_{\mu}}\}_{\mu=1}^{d^t}$; then arbitrary operator $O_{p\in \textrm{S}_t}$ is block-diagonal in matrix basis $\{\ket{\tilde e_{\mu}}\bra{\tilde e_{\nu}}\}_{\mu,\nu=1}^{d^t}$;
\end{itemize}
What is very important the block-diagonalisation of the tensor power \eqref{Ut} holds as well for any invertible complex matrix $\alpha$:
\begin{equation}
    \label{GLDecomp}
    \alpha^{\otimes t}=\sum_{i}\frac{1}{D^i_V}\sum_{m_1,m_2=1}^{D^i_L}\operatorname{Tr}\left(\alpha^{\otimes t}\hat\Pi^{m_1 m_2\dagger}_i\right)\hat\Pi^{m_1 m_2}_i.
\end{equation}
This is because the groups $\textrm{GL}(d,\mathbb C)$ of complex invertible $d\times d$ matrices and the symmetric group $\textrm{S}_t$, treated as matrix algebras acting on the tensor space $ (\mathbb C^d)^{\otimes t}$, are also commutants of each other, and all the construction of Schur vector basis for their case is entirely the same. The only difference is that the subspaces $L^i_\lambda$ are now irreducibly invariant under the collective action of the general linear group $\textrm{GL}(d,\mathbb C)$.

We can also define projectors onto the entire $i$-th irreducible subspaces:
\begin{equation}
    \label{iProj}
    \hat\Pi_i=\sum_{\lambda=1}^{D^i_V}\sum_{m=1}^{D^i_L}\hat\Pi_{ii}^{m\lambda m\lambda}.
\end{equation}
They fulfill the following property:
\begin{equation}
    \label{iProjML}
     \hat\Pi_i=\sum_{m=1}^{D^i_L}\hat\Pi_{i}^{mm}=\sum_{\lambda=1}^{D^i_V}\hat\Pi_{i}^{\lambda\lambda},
\end{equation}
which follows directly from definition \eqref{PiBasis}.

\begin{lem}
The operators \eqref{PiBasis} fulfill the following important property, which will be extensively used in our considerations:
\begin{eqnarray}
    \label{ProdPiBasisGen}
    \hat\Pi^{\lambda_1\lambda_2}_k  \hat\Pi^{m_1 m_2}_l&=&\delta_{kl}\hat\Pi^{m_1 \lambda_1m_2\lambda_2}_{kl}\nonumber\\
    \hat\Pi^{m_1 m_2}_l\hat\Pi^{\lambda_1\lambda_2}_k&=&\delta_{kl}\hat\Pi^{m_1 \lambda_1m_2\lambda_2}_{lk}
\end{eqnarray}
\end{lem}
\begin{proof}
Using definition \eqref{PiBasis} and the orthogonality relations \eqref{SBOrt} we obtain:
\begin{eqnarray}
 \hat\Pi^{\lambda_1\lambda_2}_k  \hat\Pi^{m_1 m_2}_l&=&\sum_{\lambda,m}\hat\Pi^{m\lambda_1m\lambda_2}_{kk}  \hat\Pi^{m_1\lambda m_2\lambda}_{ll}\nonumber\\
 &=&\sum_{\lambda,m}\hat\Pi^{m\lambda_1m_2\lambda}_{kl}\delta_{kl}\delta_{mm_1}\delta_{\lambda_2\lambda}\nonumber\\
 &=&\delta_{kl}\hat\Pi^{m_1\lambda_1m_2\lambda_2}_{kl}.
\end{eqnarray}
Similarly:
\begin{eqnarray}
  \hat\Pi^{m_1 m_2}_l\hat\Pi^{\lambda_1\lambda_2}_k &=&\sum_{m,\lambda}\hat\Pi^{m_1\lambda m_2\lambda}_{ll}\hat\Pi^{m\lambda_1m\lambda_2}_{kk}\nonumber\\
 &=&\sum_{m,\lambda}\hat\Pi^{m_1\lambda m\lambda_2}_{lk}\delta_{kl}\delta_{m_2m}\delta_{\lambda\lambda_1}\nonumber\\
 &=&\delta_{kl}\hat\Pi^{m_1\lambda_1m_2\lambda_2}_{lk}.
\end{eqnarray}
\end{proof}
As a corollary we obtain the following three properties:
\begin{enumerate}
    \item The operators \eqref{PiBasis} commute:
    \begin{equation}
    \label{PiCommute}
    \left[ \hat\Pi^{\lambda_1\lambda_2}_k,  \hat\Pi^{m_1 m_2}_l\right]=0.
\end{equation}
\item The product of the operators for the same irreducible subspace reads:
\begin{equation}
    \label{ProdPiBasis}
  \hat\Pi^{\lambda_1\lambda_2}_i  \hat\Pi^{m_1 m_2}_i=\hat\Pi^{m_1 \lambda_1m_2\lambda_2}_{ii}.
\end{equation}
\item Trace of their product reads:
\begin{equation}
    \label{TracePiBasis}
\operatorname{Tr}\left(  \hat\Pi^{\lambda_1\lambda_2}_k  \hat\Pi^{m_1 m_2}_l\right)=\delta_{kl}\delta_{\lambda_1\lambda_2}\delta_{m_1m_2}.
\end{equation}
\end{enumerate}
The properties \eqref{PiCommute} and\eqref{ProdPiBasis} are obvious consequences of \eqref{ProdPiBasisGen}, the last one follows from \eqref{ProdPiBasisGen} and orthonormality of basis elements in the definition \eqref{FullPiBasis}. 

As a last remark on the discussed operators let us note that the oprators $\hat\Pi_i$ play the role of the identity operator on the irreducible subspaces:
\begin{eqnarray}
\label{PiIdent}
&&\hat\Pi_{i}^{\lambda_1\lambda_2}\hat\Pi_j=\hat\Pi_j\hat\Pi_{i}^{\lambda_1\lambda_2}=\delta_{ij}\hat\Pi_{i}^{\lambda_1\lambda_2}\nonumber\\
&&\hat\Pi_{i}^{m_1m_2}\hat\Pi_j=\hat\Pi_j\hat\Pi_{i}^{m_1m_2}=\delta_{ij}\hat\Pi_{i}^{m_1m_2}.
\end{eqnarray}
The proof of the first property goes as follows, all the rest are entirely analogous:
\begin{eqnarray}
\hat\Pi_{i}^{\lambda_1\lambda_2}\hat\Pi_j=\sum_m\hat\Pi_{i}^{\lambda_1\lambda_2}\hat\Pi^{mm}_j=\delta_{ij}\sum_m\hat\Pi_{ij}^{m\lambda_1m\lambda_2}=\delta_{ij}\hat\Pi_{i}^{\lambda_1\lambda_2}.\nonumber\\
\end{eqnarray}

Decompositions \eqref{Ut} and \eqref{Op} are often in the literature presented in the \textit{virtual tensor product} form. This form is useful in the context of application of Schur's Lemma to averaging, therefore let us for completeness briefly recall it. Let us fix $i$-th irreducible subspace. Let $\mathcal I_i$ denote trivial immersion of the matrices defined on the $i$-th subspace into entire space of matrices on $\mathbb C^{d^t}$. Let $\{\pi_i^{m_1m_2}\}$ be a standard matrix basis on the space of $D^i_L\times D^i_L$ matrices. Then the operator $\hat \Pi_i^{m_1m_2}$ has the following representation in terms of the matrices $\pi_i^{m_1m_2}$:
\begin{eqnarray}
\label{PiMTensor}
\textrm{U-basis:}\,\,\,\hat \Pi_i^{m_1m_2}&=&\mathcal I_i\left(\id_{D^i_V}\otimes \pi_i^{m_1m_2}\right)\nonumber\\
\textrm{S-basis:}\,\,\,\hat \Pi_i^{m_1m_2}&=&\mathcal I_i\left(\pi_i^{m_1m_2}\otimes\id_{D^i_V} \right).
\end{eqnarray}
Similarly for any operator $\hat \Pi_i^{\lambda_1\lambda_2}$ we have:
\begin{eqnarray}
\label{PiLTensor}
\textrm{U-basis:}\,\,\,\hat \Pi_i^{\lambda_1\lambda_2}&=&\mathcal I_i\left(\pi_i^{\lambda_1\lambda_2}\otimes\id_{D^i_L} \right)\nonumber\\
\textrm{S-basis:}\,\,\,\hat \Pi_i^{\lambda_1\lambda_2}&=&\mathcal I_i\left(\id_{D^i_L}\otimes \pi_i^{\lambda_1\lambda_2}\right),
\end{eqnarray}
in which the operators $\{\pi_i^{\lambda_1\lambda_2}\}$
are elements of the standard matrix basis on the space of 
$D^i_V\times D^i_V$ matrices. Both sets of formulas \eqref{PiMTensor} and \eqref{PiLTensor} can be easily obtained by formally treating the elements of Schur basis $\ket{i,m,\lambda}$ for a given \textit{fixed} irreducible subspace $i$ as elementary tensors $\ket{m}_i\otimes \ket{\lambda}_i$, in which the vectors $\{m\}_i$ and $\{\lambda\}_i$ span \textit{virtual subspaces} of dimensions $D^i_L$ and $D^i_V$ respectively. In this formulation the joint action of the unitary and symmetric group on the space 
$(\mathbb C^d)^{\otimes t}$ can be represented as:
\begin{eqnarray}
\label{USTensor}
\textrm{U-basis:}\,\,\,&&\bigoplus_i (O_p)_i\otimes U_i\nonumber\\
\textrm{S-basis:}\,\,\,&&\bigoplus_i U_i\otimes(O_p)_i,
\end{eqnarray}
in which the matrix elements of operators $(O_p)_i$ and $U_i$ are defined as:
\begin{eqnarray}
&&\left[(O_p)_i\right]_{\lambda_1\lambda_2}=\operatorname{Tr}\left(O_p\hat\Pi_i^{\lambda_1\lambda_2\dagger}\right)\nonumber\\
&&\left[U_i\right]_{m_1m_2}=\operatorname{Tr}\left(U^{\otimes t}\hat\Pi_i^{m_1m_2\dagger}\right).
\end{eqnarray}

\subsection{Schur's Lemma and its applications for averaging}
\label{App:SLemmas}

Let us start from a more abstract formulation of the Schur's Lemma. Let $\mathcal H_1$ and $\mathcal H_2$ be two complex vector spaces, let $\pi_{\mathcal H_1}$ and $\pi_{\mathcal H_2}$ be two representations of some group $G$ on these spaces and let $f:\mathcal H_1\mapsto\mathcal H_2$ be a \textit{$G$-equivariant} map between the representation spaces, which means that the map commutes with the action of the group:
\begin{equation}
    \label{equivMap}
    \pi_{\mathcal H_2}\circ f=f\circ\pi_{\mathcal H_1}.
\end{equation}
Then the following two properties hold (Schur's Lemmas):
\begin{lem}
If $\pi_{\mathcal H_1}$ and $\pi_{\mathcal H_2}$ are irreducible representations of $G$ then:
\begin{itemize}
    \item if $\mathcal H_1$ and $\mathcal H_2$ are non-isomorphic, then the only $G$-equivariant map is zero, $f\equiv 0$;
    \item  if $\mathcal H_1$ and $\mathcal H_2$ are isomorphic and if $\pi_{\mathcal H_1}$ and $\pi_{\mathcal H_2}$ are equivalent, then the only $G$-equivariant maps are multiples of identity, $f\equiv c\id_{\mathcal H_1}$.
\end{itemize}
\end{lem}
It is worth noticing that the second part of the above lemma holds only for representations defined on vector spaces over \textit{algebraically closed fields}, as $\mathbb C$.

Let us now specify $\mathcal H_1$ and $\mathcal H_2$ as $\mathbb C^{t_1}$ and $\mathbb C^{t_2}$ respectively and let $G=\textrm{U}(d)$. Then the maps $f$ are linear operators $A: \mathbb C^{t_1}\mapsto \mathbb C^{t_2}$, the representations $\pi_{\mathcal H_1}$ and $\pi_{\mathcal H_2}$ are irreducible representations of the unitary group on respectively $\mathbb C^{t_1}$ and $\mathbb C^{t_2}$  consisting of matrices $U_1$ and $U_2$,  and the $G$-equivariance means:
\begin{equation}
    U_1A=AU_2\Longleftrightarrow  U_1A(U_2)^{-1}=A\Longleftrightarrow U_1A(U_2)^{\dagger}=A.
\end{equation}
For discussed situation Schur Lemma's take the following form:
\begin{lem}
If $U_1$ and $U_2$ are irreducible representations of $\operatorname{U}(d)$ on $\mathbb C^{t_1}$ and $\mathbb C^{t_2}$ then:
\begin{itemize}
    \item for $t_1\neq t_2$ the following implication holds:
    \begin{equation}
        \label{Schur1}
        U_1A(U_2)^{\dagger}=A \Longrightarrow A\equiv 0;
    \end{equation}
    \item  for $t_1=t_2$ and for equivalent $U_1$ and $U_2$, the following implication holds:
    \begin{equation}
        \label{Schur2}
        U_1A(U_2)^{\dagger}=A \Longrightarrow A\equiv c\id.
    \end{equation}
\end{itemize}
\end{lem}
We will now apply the above version of Schur's Lemma to averaging over unitary group. We obtain the following two further lemmas:
\begin{lem}
\label{lemAv1}
Let $\{U_k\}$ be a family of non-equivalent irreducible representations of the unitary group on some complex vector spaces $\mathcal H_k$ and let $\mathcal H=\bigoplus_k \mathcal H_k$. Let $\rho\in \mathcal H$ be arbitrary linear operator on $\mathcal H$ and let $\mathcal I_k$ be trivial immersions of operators on $\mathcal H_k$ into operators on $\mathcal H$. Then the averaging of $\rho$ over representations $\{U_k\}$ is diagonal in the index $k$, namely:
\begin{equation}
\label{Schur1Twirl}
    \sum_{k,l}\int \mathcal I_k(U_k)\rho\mathcal I_l(U^\dagger_l)\operatorname{d}U
=\sum_{k}\int \mathcal I_k(U_k)\rho\mathcal I_k(U^\dagger_k)\operatorname{d}U.
\end{equation}
\end{lem}
\begin{proof}
Let $\mathcal U_{kl}(\rho)=\int \mathcal I_k(U_k)\rho\mathcal I_l(U^\dagger_l)\operatorname{d}U$. Due to the invariance of Haar measure on unitary group the operator $\mathcal U_{kl}(\rho)$ is invariant with respect to left and right action by the unitaries corresponding to respective irreducible representations:
\begin{eqnarray}
\mathcal I_k(\tilde U_k)\mathcal U_{kl}(\rho)\mathcal I_l(\tilde U^{\dagger}_l)&=&\int \mathcal I_k(\tilde U_kU_k)\rho\mathcal I_l(U^\dagger_l\tilde U^{\dagger}_l)\operatorname{d}U\nonumber\\
&=&\mathcal U_{kl}(\rho).\nonumber\\
\end{eqnarray}
Since by assumption irreps corresponding to different values of $k$ and $l$ are non-equivalent, from \eqref{Schur1} we obtain that $\mathcal U_{kl}(\rho)=0$ for $k\neq l$, hence the sum in \eqref{Schur1Twirl} contains only terms diagonal in irreps indices.
\end{proof}

\begin{lem}
\label{lemAv2}
Let $U_k$ be an irreducible representation of the unitary group on some complex vector space $\mathcal H_k$ of dimension $D_{\mathcal H}^k$. Let $\rho\in \mathcal H_k$ be arbitrary linear operator on $\mathcal H_k$. Then the effect of averaging of $\rho$ over representation $U_k$ is proportional to the identity operator on $\mathcal H_k$:
\begin{equation}
\label{Schur1Twirl2}
    \int U_k\rho U^\dagger_k\operatorname{d}U
=\frac{1}{D_{\mathcal H}^k}\operatorname{Tr}(\rho)\id_{\mathcal H_k}.
\end{equation}
\end{lem}
\begin{proof}
Let $\mathcal U_{k}(\rho)= \int U_k\rho U^\dagger_k\operatorname{d}U$. Analogously to the previous case due to the invariance of Haar measure on unitary group the operator $\mathcal U_{k}(\rho)$ is invariant with respect to left and right action by the unitaries corresponding to respective irreducible representations:
\begin{eqnarray}
\tilde U_k\mathcal U_{k}(\rho)\tilde U^{\dagger}_k=\int \tilde U_kU_k\rho U^\dagger_k\tilde U^{\dagger}_k\operatorname{d}U=\mathcal U_{k}(\rho).\nonumber\\
\end{eqnarray}
Therefore according to Schur's Lemma \eqref{Schur2}: $$\mathcal U_{k}(\rho)=c\id_{\mathcal H_k}.$$ In order to find the constant $c$ we calculate the trace of both sides of this property:
\begin{eqnarray}
\operatorname{Tr}\left(\mathcal U_{k}(\rho)\right)&=&\operatorname{Tr}\left(c\id_{\mathcal H_k}\right)\nonumber\\
\operatorname{Tr}\left(\int U_k\rho U^\dagger_k\operatorname{d}U\right)&=&c\operatorname{Tr}\left(\id_{\mathcal H_k}\right)\nonumber\\
\int\operatorname{Tr}\left( U_k\rho U^\dagger_k\right)\operatorname{d}U&=&cD^k_{\mathcal H}\nonumber\\
\int\operatorname{Tr}\left(\rho\right)\operatorname{d}U&=&cD^k_{\mathcal H}\nonumber\\
\operatorname{Tr}\left(\rho\right)&=&cD^k_{\mathcal H},
\end{eqnarray}
hence $c=\operatorname{Tr}(\rho)/D^k_{\mathcal H}$.
\end{proof}

\subsection{Reduction of the unitary twirling map}
\label{App:UnitTwirl}

The unitary twirling map:
\begin{equation}
\label{unitTwirl}
    \mathcal T(\rho)=\int U^{\otimes t}\rho\, U^{\dagger \otimes t} \operatorname{d}U, \,U \in \operatorname{U}(d),
\end{equation}
where $\operatorname{d}U$ denotes the normalised Haar measure on $\operatorname{U}(d)$, can be presented in a simple closed form using the decomposition into irreducible representations specified in \eqref{Ut} and Schur's Lemmas. Let us first introduce concise notation for this decomposition:
\begin{eqnarray}
    \label{Ut1}
    U^{\otimes t}&=&\sum_{k}\frac{1}{D^k_V}\sum_{m_1,m_2=1}^{D^i_L}\operatorname{Tr}\left(U^{\otimes t}\hat\Pi^{m_1 m_2\dagger}_k\right)\hat\Pi^{m_1 m_2}_k\nonumber\\
    &=&\frac{1}{D^k_V}U^k_{m_1m_2}\hat\Pi^{m_1 m_2}_k,
\end{eqnarray}
in which we have $U^k_{m_1m_2}=\operatorname{Tr}\left(U^{\otimes t}\hat\Pi^{m_1 m_2\dagger}_k\right)$ and we use the Einstein summation convention for the same indices appearing at opposite positions.
In a similar way we can decompose arbitrary state $\rho$ in the Schur operator basis \eqref{FullPiBasis}:
\begin{equation}
    \rho=\rho^{ij}_{m_1\lambda_1m_2\lambda_2}\hat\Pi_{ij}^{m_1\lambda_1m_2\lambda_2},
\end{equation}
in which the summation convention also holds. Let us now rewrite the twirling map \eqref{unitTwirl} in Schur operator basis:

\begin{equation}
\label{unitTwirl1}
    \mathcal T(\rho)=\int \left(\frac{1}{D^k_V}U^k_{m_1m_2}\hat\Pi^{m_1 m_2}_k\right)\left(\rho^{ij}_{n_1\lambda_1n_2\lambda_2}\hat\Pi_{ij}^{n_1\lambda_1n_2\lambda_2}\right)\left(\frac{1}{D^l_V}(U^\dagger)^l_{r_1r_2}\hat\Pi^{r_1 r_2}_l\right) \operatorname{d}U,
\end{equation}
in which $(U^\dagger)^l_{r_1r_2}=\operatorname{Tr}\left((U^\dagger)^{\otimes t}\hat\Pi^{r_1 r_2\dagger}_l\right)$. The expression \eqref{unitTwirl1} can be simplified in three steps:
\begin{itemize}
    \item application of  Schur's Lemma for inequivalent irreducible representations \eqref{Schur1Twirl}, due to which averaging in \eqref{unitTwirl1} is block-diagonal:
    \begin{equation}
    \label{unitTwirl2}
    \mathcal T(\rho)=\int \left(\frac{1}{D^k_V}U^k_{m_1m_2}\hat\Pi^{m_1 m_2}_k\right)\left(\rho^{ij}_{n_1\lambda_1n_2\lambda_2}\hat\Pi_{ij}^{n_1\lambda_1n_2\lambda_2}\right)\left(\frac{1}{D^k_V}(U^\dagger)^k_{r_1r_2}\hat\Pi^{r_1 r_2}_k\right) \operatorname{d}U;
\end{equation}
note that in the above equation the summation convention is applied in a way that the summation indices are \textit{unbounded} and common for the entire expression; namely the transition from \eqref{unitTwirl1} to \eqref{unitTwirl2} relies on changing \textit{double} sum $\sum_{kl}$ over the indices $k,l$ corresponding to in principle different irreducible subspaces into a \textit{single} sum $\sum_k$ over the same irreducible subspaces for left and right action of the unitary group; in the following formulas we treat summation convention in the same way;
\item utilizing block-orthogonality of the corresponding operator bases:
\begin{equation}
\hat\Pi_k^{m_1m_2}\hat\Pi_{ij}^{n_1\lambda_1n_2\lambda_2}\hat\Pi_k^{r_1r_2}=\hat\Pi_k^{m_1m_2}\hat\Pi_k^{n_1n_2}\hat\Pi_k^{\lambda_1\lambda_2}\hat\Pi_k^{r_1r_2},
\end{equation}
which follows from orthogonality relations \eqref{SBOrt} and \eqref{ProdPiBasis}; due to this property we have:
 \begin{eqnarray}
    \label{unitTwirl3}
    \mathcal T(\rho)&=&\int \left(\frac{1}{D^k_V}U^k_{m_1m_2}\hat\Pi^{m_1 m_2}_k\right)\left(\rho^{kk}_{n_1\lambda_1n_2\lambda_2}\hat\Pi_{k}^{n_1n_2}\hat\Pi_{k}^{\lambda_1\lambda_2}\right)\left(\frac{1}{D^k_V}(U^\dagger)^k_{r_1r_2}\hat\Pi^{r_1 r_2}_k\right) \operatorname{d}U\nonumber\\
    &=&\rho^{kk}_{n_1\lambda_1n_2\lambda_2}\hat\Pi_{k}^{\lambda_1\lambda_2}\int \left(\frac{1}{D^k_V}U^k_{m_1m_2}\hat\Pi^{m_1 m_2}_k\right)\hat\Pi_{k}^{n_1n_2}\left(\frac{1}{D^k_V}(U^\dagger)^k_{r_1r_2}\hat\Pi^{r_1 r_2}_k\right) \operatorname{d}U,
\end{eqnarray}
in which the equality follows from \textit{commutativity} of the operators $\hat\Pi_k^{\lambda_1\lambda_2}$ and $\hat\Pi_k^{n_1n_2}$;
\item application of Schur's Lemma for averaging over equivalent irreducible representations \eqref{Schur1Twirl2}, due to which the integral reads:
\begin{equation}
    \int \left(\frac{1}{D^k_V}U^k_{m_1m_2}\hat\Pi^{m_1 m_2}_k\right)\hat\Pi_{k}^{n_1n_2}\left(\frac{1}{D^k_V}(U^\dagger)^k_{r_1r_2}\hat\Pi^{r_1 r_2}_k\right) \operatorname{d}U=\frac{1}{D^k_L}\hat \Pi_k\delta_k^{n_1n_2},
\end{equation}
in which the tensor $\delta_k^{n_1n_2}$ represents a Kronecker delta on the $k$-th irreducible subspace,
and therefore we have:
\begin{eqnarray}
\label{unitTwirl4}
    \mathcal T(\rho)&=&\frac{1}{D^k_L}\delta_k^{n_1n_2}\rho^{kk}_{n_1\lambda_1n_2\lambda_2}\hat\Pi_{k}\hat\Pi_{k}^{\lambda_1\lambda_2}=\frac{1}{D^k_L}\rho^{k}_{\lambda_1\lambda_2}\hat\Pi_{k}^{\lambda_1\lambda_2},
\end{eqnarray}
in which we used simplified notation: $\rho^{k}_{\lambda_1\lambda_2}=\operatorname{Tr}\left(\rho\hat\Pi_{k}^{\lambda_1\lambda_2\dagger}\right)$ and used the fact that $\hat \Pi_k$ plays the role of an identity on the $k$-th subspace, see \eqref{PiIdent}; the last equality can be easilly proved by the following chain of identities:
\begin{equation}
    \delta_k^{n_1n_2}\rho^{kk}_{n_1\lambda_1n_2\lambda_2}=\sum_{n_1,n_2=1}^{D^k_L} \delta_k^{n_1n_2}\operatorname{Tr}\left(\rho\hat\Pi_{kk}^{n_1\lambda_1n_2\lambda_2\dagger}\right)=\sum_{n=1}^{D^k_L}\operatorname{Tr}\left(\rho\hat\Pi_{kk}^{n\lambda_1n\lambda_2\dagger}\right)=\operatorname{Tr}\left(\rho\left(\sum_{n=1}^{D^k_L}\hat\Pi_{kk}^{n\lambda_1n\lambda_2\dagger}\right)\right)=\operatorname{Tr}\left(\rho\hat\Pi_{k}^{\lambda_1\lambda_2\dagger}\right),
\end{equation}
in which in the last step we used definition of the $\hat\Pi_k^{\lambda_1\lambda_2}$ operators \eqref{PiBasis}.
\end{itemize}
To sum up, the final form of the twirling map is given by the following expression:
\begin{equation}
    \label{fullUnitTwirlMap}
    \mathcal T(\rho)=\sum_{k}\frac{1}{D^k_L}\sum_{\lambda_1\lambda_2=1}^{D^k_V}\operatorname{Tr}\left(\rho\hat\Pi_{k}^{\lambda_1\lambda_2\dagger}\right)\hat\Pi_{k}^{\lambda_1\lambda_2}.
\end{equation}
It is equivalent to a permutation operator \eqref{Op}. The map $\mathcal T(\rho)$ is idempotent, $\mathcal T(\mathcal T(\rho))=\mathcal T(\rho)$ (see below for a proof), hence a state of the block-diagonal form:
\begin{equation}
    \label{rhoInvariant}
    \rho_{\textrm{U}}=\sum_{k}\frac{1}{D^k_L}\sum_{\lambda_1\lambda_2=1}^{D^k_V}\rho^{k}_{\lambda_1\lambda_2}\hat\Pi_{k}^{\lambda_1\lambda_2}
\end{equation}
is invariant under unitary twirling operation: $\mathcal T( \rho_{\textrm{U}})= \rho_{\textrm{U}}$. The $(D^k_V\times D^k_V)$-dimensional subsystems spanned by $\{\hat\Pi_k^{\lambda_1\lambda_2}\}$ are called \textit{noiseless subsystems} or \textit{decoherence-free subsystems}. The idempotence of the unitary twirling map follows simply from the invariance of the map with respect to the action of the unitary operations:
\begin{eqnarray}
\tilde U_k\mathcal T(\rho)\tilde U^{\dagger}_k=\int \tilde U_kU_k\rho U^\dagger_k\tilde U^{\dagger}_k\operatorname{d}U=\mathcal T(\rho).\nonumber\\
\end{eqnarray}

\subsection{Reduction of the symmetric twirling map}
\label{App:SymTwirl}
Derivation of a closed form \eqref{fullSymTwirlMap0} of a symmetric twirling map \eqref{SymTwirlDef0} is performed in full analogy with the one for the unitary twirling map. First, note that Schur's Lemmas \ref{lemAv1} and \ref{lemAv2} for unitary averaging  hold in entirely analogous form for symmetric averaging:
\begin{lem}
\label{lemAvSym1}
Let $\{O_k\}$ be a family of non-equivalent irreducible representations of the symmetric group $\operatorname{S}_t$ on some complex vector spaces $\mathcal H_k$ and let $\mathcal H=\bigoplus_k \mathcal H_k$. Let $\rho\in \mathcal H$ be arbitrary linear operator on $\mathcal H$ and let $\mathcal I_k$ be trivial immersions of operators on $\mathcal H_k$ into operators on $\mathcal H$. Then the averaging of $\rho$ over representations $\{O_k\}$ is diagonal in the index $k$, namely:
\begin{eqnarray}
\label{Schur1TwirlSym}
    &&\sum_{k,l}\sum_p\left(\frac{1}{t!}\right) \mathcal I_k\left((O_p)_k\right)\rho\,\mathcal I_l\left((O_p^T)_l\right)\nonumber\\
&&=\sum_{k}\sum_p\left(\frac{1}{t!}\right) \mathcal I_k\left((O_p)_k\right)\rho\,\mathcal I_k\left((O_p^T)_k\right),
\end{eqnarray}
in which by $(O_p)_k$ we mean the value of the $k$-th irreducible representation when applied to a permutation $p$.
\end{lem}
\begin{proof}
Let $\mathcal O_{kl}(\rho)=\sum_p\left(\frac{1}{t!}\right) \mathcal I_k\left((O_p)_k\right)\rho\,\mathcal I_l\left((O_p^T)_l\right)$.  $\mathcal O_{kl}(\rho)$ is invariant with respect to left and right action by the permutations corresponding to respective irreducible representations:
\begin{eqnarray}
&&\mathcal I_k((O_q)_k)\mathcal O_{kl}(\rho)\mathcal I_l\left((O_q^T)_l\right)\nonumber\\
&&=\sum_p\left(\frac{1}{t!}\right) \mathcal I_k((O_q)_k(O_p)_k)\rho\,\mathcal I_l\left((O_p^T)_l(O_q^T)_l\right)\nonumber\\
&&=\sum_p\left(\frac{1}{t!}\right) \mathcal I_k((O_{qp})_k)\rho\,\mathcal I_l\left((O_{qp}^T)_l\right)\nonumber\\
&&=\sum_{qp}\left(\frac{1}{t!}\right) \mathcal I_k((O_{qp})_k)\rho\,\mathcal I_l\left((O_{qp}^T)_l\right)=\nonumber\\
&&=\sum_r\left(\frac{1}{t!}\right) \mathcal I_k((O_{r})_k)\rho\,\mathcal I_l\left((O_{r}^T)_l\right)=\mathcal O_{kl}(\rho).
\end{eqnarray}
Since by assumption irreps corresponding to different values of $k$ and $l$ are non-equivalent, from \eqref{Schur1} we obtain that $\mathcal O_{kl}(\rho)=0$ for $k\neq l$, hence the sum in \eqref{Schur1Twirl} contains only terms diagonal in irreps indices.
\end{proof}
\begin{lem}
\label{lemAvSym2}
Let $O_k$ be an irreducible representation of the symmetric group $\operatorname{S}_t$  on some complex vector space $\mathcal H_k$ of dimension $D_{\mathcal H}^k$. Let $\rho\in \mathcal H_k$ be arbitrary linear operator on $\mathcal H_k$. Then the effect of averaging of $\rho$ over representation $O_k$ is proportional to the identity operator on $\mathcal H_k$:
\begin{equation}
\label{Schur1TwirlSym2}
   \sum_p\left(\frac{1}{t!}\right) (O_p)_k\rho\,(O_p^T)_k
=\frac{1}{D_{\mathcal H}^k}\operatorname{Tr}(\rho)\id_{\mathcal H_k}.
\end{equation}
\end{lem}
\begin{proof}
Let $\mathcal O_{k}(\rho)=\sum_p\left(\frac{1}{t!}\right) (O_p)_k\rho\,(O_p^T)_k$. Analogously to the previous case  the operator $\mathcal O_{k}(\rho)$ is invariant with respect to left and right action by the permutations corresponding to respective irreducible representations:
\begin{eqnarray}
(O_q)_k\mathcal O_{k}(\rho)(O_q^T)_k&=& \sum_p\left(\frac{1}{t!}\right) (O_q)_k(O_p)_k\rho\,(O_p^T)_k(O_q^T)_k\nonumber\\
&=&\sum_p\left(\frac{1}{t!}\right) (O_{qp})_k\rho\,(O_{qp}^T)_k=\mathcal O_{k}(\rho)\nonumber\\
\end{eqnarray}
Therefore according to Schur's Lemma \eqref{Schur2}: $$\mathcal O_{k}(\rho)=c\id_{\mathcal H_k}.$$ In order to find the constant $c$ we calculate the trace of both sides of this property:
\begin{eqnarray}
\operatorname{Tr}\left(O_{k}(\rho)\right)&=&\operatorname{Tr}\left(c\id_{\mathcal H_k}\right)\nonumber\\
\operatorname{Tr}\left(\sum_p\left(\frac{1}{t!}\right) (O_p)_k\rho\,(O_p^T)_k\right)&=&c\operatorname{Tr}\left(\id_{\mathcal H_k}\right)\nonumber\\
\sum_p\left(\frac{1}{t!}\right) \operatorname{Tr}\left((O_p)_k\rho\,(O_p^T)_k\right)&=&cD^k_{\mathcal H}\nonumber\\
\sum_p\left(\frac{1}{t!}\right)\operatorname{Tr}\left(\rho\right)&=&cD^k_{\mathcal H}\nonumber\\
\operatorname{Tr}\left(\rho\right)&=&cD^k_{\mathcal H},
\end{eqnarray}
hence $c=\operatorname{Tr}(\rho)/D^k_{\mathcal H}$.
\end{proof}
All the steps of the derivation of closed form of symmetric twirling map are analogous to the ones for unitary twirling, we however present them below for completeness. Utilising the relation (expressed for convenience in Einstein summation convention):
\begin{eqnarray}
    O_{p}&=&\sum_{k}\frac{1}{D^k_L}\sum_{\lambda_1,\lambda_2=1}^{D^k_V}\operatorname{Tr}\left(O_p\hat\Pi^{\lambda_1 \lambda_2\dagger}_k\right)\hat\Pi^{\lambda_1 \lambda_2}_k\nonumber\\
    &=&\frac{1}{D^k_L}\left(O_p\right)^k_{\lambda_1\lambda_2}\hat\Pi^{\lambda_1 \lambda_2}_k.\nonumber
\end{eqnarray}
we express the symmetric twirling map in full analogy with the expression \eqref{unitTwirl1} for unitary twirling:

\begin{equation}
\label{symTwirl1}
    \mathcal T_{\textrm{sym}}(\rho)=\sum_p\left(\frac{1}{t!}\right) \left(\frac{1}{D^k_L}\left(O_p\right)^k_{\mu_1\mu_2}\hat\Pi^{\mu_1 \mu_2}_k\right)\left(\rho^{ij}_{n_1\lambda_1n_2\lambda_2}\hat\Pi_{ij}^{n_1\lambda_1n_2\lambda_2}\right)\left(\frac{1}{D^l_L}\left(O_p^T\right)^l_{\eta_1\eta_2}\hat\Pi^{\eta_1 \eta_2}_l\right).
\end{equation}
 The expression \eqref{symTwirl1} is simplified analogously:
\begin{itemize}
    \item application of  Schur's Lemma for inequivalent irreducible representations \eqref{Schur1TwirlSym}, due to which averaging in \eqref{symTwirl1} is block-diagonal:
   \begin{equation}
\label{symTwirl2}
    \mathcal T_{\textrm{sym}}(\rho)=\sum_p\left(\frac{1}{t!}\right) \left(\frac{1}{D^k_L}\left(O_p\right)^k_{\mu_1\mu_2}\hat\Pi^{\mu_1 \mu_2}_k\right)\left(\rho^{ij}_{n_1\lambda_1n_2\lambda_2}\hat\Pi_{ij}^{n_1\lambda_1n_2\lambda_2}\right)\left(\frac{1}{D^k_L}\left(O_p^T\right)^k_{\eta_1\eta_2}\hat\Pi^{\eta_1 \eta_2}_k\right);
\end{equation}
as before the entire above expression is a \textit{single} sum $\sum_k$ over the same irreducible subspaces for left and right action of the symmetric group; 
\item utilizing block-orthogonality of the corresponding operator bases:
\begin{equation}
\hat\Pi_k^{\mu_1\mu_2}\hat\Pi_{ij}^{n_1\lambda_1n_2\lambda_2}\hat\Pi_k^{\eta_1\eta_2}=\hat\Pi_k^{\mu_1\mu_2}\hat\Pi_k^{n_1n_2}\hat\Pi_k^{\lambda_1\lambda_2}\hat\Pi_k^{\eta_1\eta_2},
\end{equation}
which follows from orthogonality relations \eqref{SBOrt} and \eqref{ProdPiBasis}; due to this property we have:
 \begin{eqnarray}
    \label{symTwirl3}
   \mathcal T_{\textrm{sym}}(\rho)&=&\sum_p\left(\frac{1}{t!}\right) \left(\frac{1}{D^k_L}\left(O_p\right)^k_{\mu_1\mu_2}\hat\Pi^{\mu_1\mu_2}_k\right)\left(\rho^{kk}_{n_1\lambda_1n_2\lambda_2}\hat\Pi_{k}^{n_1n_2}\hat\Pi_{k}^{\lambda_1\lambda_2}\right)\left(\frac{1}{D^k_L}\left(O_p^T\right)^k_{\eta_1\eta_2}\hat\Pi^{\eta_1 \eta_2}_k\right)\nonumber\\
    &=&\rho^{kk}_{n_1\lambda_1n_2\lambda_2}\hat\Pi_{k}^{n_1n_2}\sum_p\left(\frac{1}{t!}\right) \left(\frac{1}{D^k_L}\left(O_p\right)^k_{\mu_1\mu_2}\hat\Pi^{\mu_1\mu_2}_k\right)\hat\Pi_{k}^{\lambda_1\lambda_2}\left(\frac{1}{D^k_L}\left(O_p^T\right)^k_{\eta_1\eta_2}\hat\Pi^{\eta_1 \eta_2}_k\right),
\end{eqnarray}
in which the equality follows from \textit{commutativity} of the operators $\hat\Pi_k^{\mu_1\mu_2}$ and $\hat\Pi_k^{n_1n_2}$; here is the crucial difference with respect to derivation of the unitary twirling: the averaging has been reduced to averaging of Schur operator basis elements spanning the irreducible representations of the symmetric group;
\item application of Schur's Lemma for averaging over equivalent irreducible representations \eqref{Schur1TwirlSym2}, due to which the sum over all permutations $p\in \textrm S_t$ reads:
\begin{equation}
   \sum_p\left(\frac{1}{t!}\right) \left(\frac{1}{D^k_L}\left(O_p\right)^k_{\mu_1\mu_2}\hat\Pi^{\mu_1\mu_2}_k\right)\hat\Pi_{k}^{\lambda_1\lambda_2}\left(\frac{1}{D^k_L}\left(O_p^T\right)^k_{\eta_1\eta_2}\hat\Pi^{\eta_1 \eta_2}_k\right)=\frac{1}{D^k_V}\hat \Pi_k\delta_k^{\lambda_1\lambda_2};
\end{equation}
therefore we finally have:
\begin{eqnarray}
\label{symTwirl4}
    \mathcal T_{\textrm{sym}}(\rho)&=&\frac{1}{D^k_V}\delta_k^{\lambda_1\lambda_2}\rho^{kk}_{n_1\lambda_1n_2\lambda_2}\hat\Pi_{k}\hat\Pi_{k}^{n_1n_2}=\frac{1}{D^k_V}\rho^{k}_{n_1n_2}\hat\Pi_{k}^{n_1n_2}=\sum_{k}\frac{1}{D^k_V}\sum_{n_1,n_2=1}^{D^k_L}\operatorname{Tr}\left(\rho\hat\Pi^{n_1 n_2\dagger}_k\right)\hat\Pi^{n_1 n_2}_k,
\end{eqnarray}
in which we used simplified notation: $\rho^{k}_{n_1n_2}=\operatorname{Tr}\left(\rho\hat\Pi_{k}^{n_1n_2\dagger}\right)$ and used the fact that $\hat \Pi_k$ plays the role of an identity on the $k$-th subspace, see \eqref{PiIdent}; the last equality can be easilly proved by the following chain of identities:
\begin{equation}
    \delta_k^{\lambda_1\lambda_2}\rho^{kk}_{n_1\lambda_1n_2\lambda_2}=\sum_{\lambda_1\lambda_2=1}^{D^k_V} \delta_k^{\lambda_1\lambda_2}\operatorname{Tr}\left(\rho\hat\Pi_{kk}^{n_1\lambda_1n_2\lambda_2\dagger}\right)=\sum_{\lambda=1}^{D^k_V}\operatorname{Tr}\left(\rho\hat\Pi_{kk}^{n_1\lambda n_2\lambda\dagger}\right)=\operatorname{Tr}\left(\rho\left(\sum_{\lambda=1}^{D^k_V}\hat\Pi_{kk}^{n_1\lambda n_2\lambda\dagger}\right)\right)=\operatorname{Tr}\left(\rho\hat\Pi_{k}^{n_1n_2\dagger}\right),
\end{equation}
in which in the last step we used definition of the $\hat\Pi_k^{n_1n_2}$ operators \eqref{PiBasis}. 
\end{itemize}
The symmetric twirling operation \eqref{symTwirl4} is also idempotent (due to invariance of the map with respect to the action of the permutation operators), and therefore the states of the form:
\begin{equation}
      \rho_\textrm{S}=\sum_{k}\frac{1}{D^k_V}\sum_{m_1,m_2=1}^{D^k_L}\rho_{m_1 m_2}^k\hat\Pi^{m_1 m_2}_k,
\end{equation}
are untouched by this map: $T_{\textrm{sym}}(\rho_\textrm{S})=\rho_\textrm{S}$, and operators $\hat\Pi_k^{m_1m_2}$ span \textit{noiseless subsystems} with respect to symmetric twirling.

\section{Some proofs connected with SLOCC twirling map}
\subsection{Proof of Lemma \ref{lem:trace}}
\label{App:LemTrace}
Here we provide a direct proof that the SLOCC map $\mathcal S (\rho)$ \eqref{SLOCCMap} is trace non-increasing.
\begin{eqnarray}
\mathrm{Tr} (\mathcal S (\rho)) & = &  \mathrm{Tr} \left( \bigotimes_{i=1}^t L_i\rho\bigotimes_{i=1}^t L_i^{\dagger} \right) \nonumber \\
& = &  \mathrm{Tr} \left( \bigotimes_{i=1}^t L_i^{\dagger} \bigotimes_{i=1}^t L_i \rho \right)  \nonumber \\
& = &  \mathrm{Tr} \left( \left( \bigotimes_{i=1}^t  L_i^{\dagger} L_i  \right) \rho \right),
\end{eqnarray}
in which $L_i$ denotes normalised special linear matrix, $L_i=M_i/\|M_i\|$ for $M_i\in \textrm{SL}(d,\mathbb C)$.
Due to the Cartan decomposition (which is in the case of $\textrm{SL}(d, \mathbb C)$ equivalent to the SVD-decomposition) each matrix $M_i$ can be expressed as a product $K_i' A_i K_i$ where $K_i, K_i'$ are unitary and $A_i$ is diagonal with its largest entry denoted $x_i$ \cite{Markiewicz21}. Moreover we have $\| A_i \| = x_i$, which implies, that:
$$L_i^\dagger L_i = \frac{M_i^\dagger M_i}{\|M_i\|^2}  = \frac{1}{x_i^2} K_i^\dagger A_i^2 K_i. $$
Further we have:
\begin{eqnarray}
\mathrm{Tr} (\mathcal S (\rho)) & = &  \mathrm{Tr} \left( \left( \bigotimes_{i=1}^t  \frac{1}{x_i^2} K_i^\dagger A_i^2 K_i \right) \rho \right) \nonumber \\
& = &   \mathrm{Tr} \left( \bigotimes_{i=1}^t A_i^2 \bigotimes_{i=1}^t K_i \rho \bigotimes_{i=1}^t K_i^\dagger   \right) \prod_{i = 1}^t x_i^{-2} \nonumber \\
& = &  \mathrm{Tr} \left( \left(   \bigotimes_{i=1}^t A_i^2 \right) \tilde \rho \right)  \prod_{i = 1}^t x_i^{-2},
\end{eqnarray}
in which a unitarily evolved state $\tilde\rho$ reads:
\begin{equation}
\tilde \rho = \bigotimes_{i=1}^t K_i \rho \bigotimes_{i=1}^t K_i^\dagger 
\end{equation}
The matrix $A = \bigotimes_{i=1}^t A_i^2$ is diagonal with $\| A \| =  \prod_{i = 1}^t x_i^{2}$. Therefore:
$$ \mathrm{Tr} \left( \left(   \bigotimes_{i=1}^t A_i^2 \right) \tilde \rho \right)  \prod_{i = 1}^t x_i^{-2} = \sum_{i = 1}^t \frac{A_{i i} }{\|A\|} \tilde \rho_{i i}.$$
Knowing that $A_{i i} / \|A\| \leq 1$, and from the fact that each positive matrix $\rho$ has non-negative elements on the diagonal we obtain:
$$ \mathrm{Tr} (\mathcal S (\rho)) =   \sum_{i = 1}^t \frac{A_{i i} }{\|A\| } \tilde \rho_{i i} \leq  \sum_{i = 1}^t  \tilde \rho_{i i} = \mathrm{Tr} (\tilde \rho) = \mathrm{Tr} (\rho).$$

\textbf{Remark.} Note that the above proof does not depend on the internal structure of the matrices $A_i$ apart from the assumption that they are diagonal matrices, which holds also in the case of a Cartan decomposition of any reductive Lie group (matrices $A$ represent maximal abelian subgroup, therefore we can always choose a basis in which the representation is diagonal). Therefore the above proof holds also for the case of \textit{general stochastic operation} defined in formula \eqref{gsodef}.

\subsection{Proof of Lemma \ref{lem2}}
\label{proof:lem2}

Inserting the decomposition of $\alpha^{\otimes t}$ \eqref{alphaBlockDiag} into the general formula for unitary twirling \eqref{fullUnitTwirlMap0}  we obtain:
\begin{eqnarray}
&&\mathcal T(\alpha^{\otimes t})\nonumber\\
&&=\sum_k\frac{1}{D^k_L}\sum_{\lambda_1,\lambda_2}\Tr\left(\alpha^{\otimes t}\hat\Pi^{\lambda_1\lambda_2\dagger}_k\right)\hat\Pi^{\lambda_1\lambda_2}_k\nonumber\\
&&=\sum_k\frac{1}{D^k_L}\sum_{\lambda_1,\lambda_2}\Tr\left(\sum_{l,m_1,m_2}\frac{1}{D^l_V}\alpha^l_{m_1m_2}\hat\Pi^{m_1m_2}_l\hat\Pi^{\lambda_1\lambda_2\dagger}_k\right)\hat\Pi^{\lambda_1\lambda_2}_k\nonumber\\
&&=\sum_k\frac{1}{D^k_L}\sum_{\lambda_1,\lambda_2}\sum_{l,m_1,m_2}\frac{1}{D^l_V}\alpha^l_{m_1m_2}\Tr\left(\hat\Pi^{m_1m_2}_l\hat\Pi^{\lambda_1\lambda_2\dagger}_k\right)\hat\Pi^{\lambda_1\lambda_2}_k\nonumber\\
&&=\sum_k\frac{1}{D^k_L}\sum_{\lambda_1,\lambda_2}\sum_{l,m_1,m_2}\frac{1}{D^l_V}\alpha^l_{m_1m_2}\delta_{kl}\delta_{m_1m_2}\delta_{\lambda_1\lambda_2}\hat\Pi^{\lambda_1\lambda_2}_k\nonumber\\
&&=\sum_k\frac{1}{D^k_LD^k_V}\sum_{\lambda,m}\alpha^k_{mm}\hat\Pi^{\lambda\lambda}_k=\sum_k\frac{1}{D^k}\sum_{\lambda,m}\alpha^k_{mm}\hat\Pi^{\lambda\lambda}_k,
\end{eqnarray}
in which we used the orthogonality property for operators \eqref{TracePiBasis}. Finally, utilizing the property \eqref{iProjML} we obtain:
\begin{eqnarray}
\mathcal T(\alpha^{\otimes t})&=&\sum_k\frac{1}{D^k}\sum_{\lambda,m}\alpha^k_{mm}\hat\Pi^{\lambda\lambda}_k\nonumber\\
&=&\sum_k\frac{1}{D^k}\sum_{\lambda,m}\Tr\left(\alpha^{\otimes t}\hat\Pi^{mm\dagger}_k\right)\hat\Pi^{\lambda\lambda}_k\nonumber\\
&=&\sum_k\frac{1}{D^k}\Tr\left(\alpha^{\otimes t}\sum_{m}\hat\Pi^{mm\dagger}_k\right)\sum_{\lambda}\hat\Pi^{\lambda\lambda}_k\nonumber\\
&=&\sum_k\frac{1}{D^k}\Tr\left(\alpha^{\otimes t}\hat\Pi_k^\dagger\right)\hat\Pi_k.
\end{eqnarray}
\end{widetext}

\end{document}